%% file: Vector_Consensus_Main.tex
\documentclass[10pt,conference,letterpaper]{IEEEtran}
\IEEEoverridecommandlockouts
\def\BibTeX{{\rm B\kern-.05em{\sc i\kern-.025em b}\kern-.08em
    T\kern-.1667em\lower.7ex\hbox{E}\kern-.125emX}}
\usepackage{algorithmicx}
\usepackage{algpseudocode}
\usepackage{algorithm}

\usepackage[shortlabels]{enumitem}
\usepackage{amsmath} 
\usepackage{amsthm}
\usepackage{amssymb}
\usepackage{mathtools} 
\usepackage{mathtools}
 
\usepackage{float}
\usepackage{subfigure,color}
\DeclarePairedDelimiter\abs{\lvert}{\rvert}

\newtheorem{theorem}{Theorem}

\newtheorem{lemma}{Lemma}
\newtheorem{corollary}{Corollary}
\newtheorem{definition}{Definition}
\newtheorem{example}{Example}

\newtheorem{remark}{Remark}

\algrenewcommand{\algorithmiccomment}[1]{\hfill// #1}

\begin{document}
	\title{Asynchronous Vector Consensus over Matrix-Weighted Networks}
	
\author{\IEEEauthorblockN{P Raghavendra Rao}
	\IEEEauthorblockA{Department of Electrical Engineering \\
		Indian Institute of Technology Tirupati, India\\
		ee20d002@iittp.ac.in}
	\and
	\IEEEauthorblockN{Pooja Vyavahare}
	\IEEEauthorblockA{Department of Electrical Engineering \\
		Indian Institute of Technology Tirupati, India\\
		poojav@iittp.ac.in}
}

\maketitle
\begin{abstract}
We study the distributed consensus of state vectors in a discrete-time multi-agent network with matrix edge weights using stochastic matrix convergence theory. We present a distributed asynchronous time update model wherein one randomly selected agent updates its state vector at a time by interacting with its neighbors. We prove that all agents converge to same state vector almost surely when every edge weight matrix is positive definite. We study vector consensus in cooperative-competitive networks with edge weights being either positive or negative definite matrices and present a necessary and sufficient condition to achieve bipartite vector consensus in such networks. We study the network structures on which agents achieve zero consensus. We also present a convergence result on nonhomogenous matrix products which is of independent interest in matrix convergence theory. All the results hold true for the synchronous time update model as well in which all agents update their states simultaneously. 
\end{abstract}


\section{Introduction}
\label{sec:intro}
Many real world systems, such as autonomous vehicles, mobile robots, and smart power grids, are described by Multi-Agent Systems (MAS) \cite{dorri2018multi}. MAS is a set of agents which communicate with each other. The communication among agents is generally represented by a directed graph. Opinion dynamics in MAS is the study of networked systems to understand the evolution of opinions (also called as states) of agents over time. One of the key challenges for agents in MAS is to reach to a common opinion, also called as \textit{distributed consensus,} by sharing limited local information. Consensus is one of the primitives of distributed computation with wide range of applications in distributed optimization \cite{nedic2010constrained}, \cite{lin2016distributed}, state estimation \cite{mitra2020new}, \cite{lalitha2018social}, robotics \cite{ota2006multi}, \cite{inigo2012robotics}, \cite{soriano2013multi} and social networks \cite{jiang2013understanding}. 

While distributed consensus algorithms for scalar states with scalar edge weights has been studied from long time in the literature (see \cite{Qin17,Ning23} for a literature survey), many MAS applications require consensus on multidimensional state vectors. For example, an autonomous vehicle equipped with multiple sensors like camera, radar etc. can be modeled as a MAS with sensors being agents and communication links between them as the underlying directed graph. At any time $t,$ sensors collect data from the environment to estimate the vehicle's attributes like velocity and position. These attributes depend on each other and hence can be modeled as a multidimensional state vector of the system. All the sensors need to agree on the same state vector for the vehicle to operate smoothly. To this end, the sensors communicate with each other and calculate the state vector by combining weighted information received from other sensors. As the attributes are interdependent and each sensor can have different accuracy levels on its estimates, the state vector of each sensor is weighted by a matrix to obtain the aggregated state vector of the system. Analysis of this model can be done by understanding the distributed consensus of vector states on matrix-weighted networks.

With these emerging applications, matrix-weighted MAS for vector consensus has gained interest in recent times \cite{Trinh18,Tran21,Pan22,Liu23,Foight20,Le-Phan24}. While most of the literature \cite{Tran21,Pan22,Liu23,Foight20} consider the MAS systems in which all agents update their state vectors simultaneously (also known as \textit{synchronous} time systems), \cite{Le-Phan24} study the state evolution when agents update states \textit{asynchronously}. The existing literature on matrix-weighted MAS provides network structure conditions for obtaining different types of consensus in MAS using properties of the graph Laplacian matrix. Stochastic matrix convergence theory has been extensively used to study scalar consensus problem \cite{ren2005consensus} but its relation to vector consensus over matrix-weighted networks is not yet explored.  Our focus in this work is on understanding the vector consensus in asynchronous time systems using matrix convergence theory and we show the almost sure convergence of vector consensus. 

\subsection{Related Work}
\label{sec:review}

Many opinion dynamics models on MAS \cite{Dong18} have been studied in the literature starting from DeGroot model \cite{DeGroot74}. DeGroot based consensus model, in which agents update their states by linearly combining states of their neighbors, is one of the most commonly studied model \cite{Liu23, olfati2007consensus, ren2005consensus}. When states are scalars, \cite{olfati2007consensus} showed that consensus is guaranteed in both continuous-time (CT) and discrete-time (DT) systems if the directed graph is strongly connected and \cite{ren2005consensus} showed the existence of spanning tree in the graph being sufficient to achieve consensus. Both these models consider a synchronous update model in which all agents in the MAS update their states simultaneously. In general MAS, only  a set of agents update their states at a time which can be picked randomly leading to an underlying stochastic model on the communication pattern. Such randomized consensus models for scalar states are studied in \cite{Boyd06, Shi15}. In some MAS, the agents update their states based on an event and in such networks distributed scalar consensus is studied extensively; see \cite{Qin17} and references therein.

  While all these results assume cooperative MAS in which all agents trust each other, in practice some agents may have mistrust on the information received from their neighbors in the network. For example, in a social network some people always disagree with opinions of a group of people in the network; see \cite{easley2010networks} for more examples. Inspired by these applications, consensus in cooperative-competitive MAS (with positive and negative edge weights representing cooperation and competition respectively) has been studied in \cite{altafini2012consensus}. Authors in \cite{altafini2012consensus} showed that in structurally balanced networks\footnote{A network is said to be \textit{structurally balanced} \cite{altafini2012consensus} if there exists a bipartition of agents (called \textit{clusters}) such that edge weights within agents of a cluster are positive and negative between the agents of two clusters.} agents of a cluster reach to a state which is negative of the state reached by agents of the other cluster and \cite{Liu17} showed that all agents achieve \textit{zero consensus} if the network is structurally unbalanced.

In all the aforementioned discussion, consensus is studied on a single attribute represented by a scalar state. Consensus on multiple attributes represented by vector states on matrix-weighted networks with edge weight being a non-negative definite matrix has been studied for CT systems in \cite{Trinh18} and for DT systems in \cite{Tran21}. Authors in \cite{Su23,Miao23} studied the consensus on matrix-weighted networks for CT systems as well as when the system switches between CT and DT. In \cite{Trinh18}, a necessary and sufficient condition was presented to achieve vector consensus on static undirected matrix-weighted networks and extended for time-varying networks in \cite{Tran21}, event triggered networks in \cite{Liu23,Liang24}, and random networks in \cite{Le-Phan24}. 

Vector consensus has also been extended to the cooperative-competitive networks \cite{Pan19,Su20,Pan22} in which cooperative agents have non-negative definite edge weight matrices between them and competitive agents have non-positive definite matrices. Bipartite consensus results and network conditions are studied in \cite{Pan19,Su20,Pan22} which show that the existence of a positive-negative spanning tree\footnote{A positive-negative spanning tree is a spanning tree in which each edge weight is either positive or negative definite matrix \cite{Tran21}.} is necessary for achieving bipartite consensus. All these works study the properties of null space of Laplacian matrix of the network to obtain the consensus or bipartite consensus results. Authors in \cite{Foight20} study the $\mathcal{H}_2$ norm of the networks to study the dynamics of consensus in matrix-weighted networks.

Our aim is to analyze vector consensus in a matrix-weighted network with the help of stochastic random matrix convergence theory. In particular, we study the vector consensus over matrix-weighted network when only one randomly chosen agent updates its state at a time. Closest to our model is the system model studied in \cite{Le-Phan24} which studies the asynchronous time systems in which two randomly selected agents updating their states by exchanging information at a time. While results of \cite{Le-Phan24} give convergence in expectation (using random matrix convergence theory), we present the almost sure convergence results. In doing so, we also derive a result on nonhomogenous matrix products which can be of independent interest.

\subsection{Main contributions}
\label{sec:contibutions}
The contributions of this paper are as follows:
\begin{enumerate}
	\item We study the discrete-time matrix-weighted MAS for distributed vector consensus problem in asynchronous update model using stochastic matrix convergence theory. We show that global consensus is achieved asymptotically almost surely when all edges have positive definite weight matrices (Theorem \ref{th:random}). We also prove a result (Lemma~\ref{lm:cauchy}) on the convergence of infinite products of matrices of the form $(A+B)$ which can be of independent interest in nonhomogeneous matrix product theory \cite{Hartfiel02}. To the best of our knowledge, this is the first work on vector consensus using matrix product convergence theory.
	\item We extend the results to cooperative-competitive networks showing bipartite vector consensus. We provide a necessary (Theorem~\ref{th:bipartite}) and sufficient (Theorem~\ref{thm:unbalanced} and Corollary \ref{cor:necessary_sufficient}) network condition to get bipartite vector consensus and present a simple proof for the sufficient condition by relating bipartite to the global consensus over network. We also analyze the network for zero consensus.
\end{enumerate} 

The remainder of this paper is organized as follows: In Section~\ref{sec:model}, we introduce the system model, along with the synchronous and asynchronous state update rules. Section~\ref{sec:matrix_convergence} presents key results in matrix convergence theory which are later applied to prove the global consensus results in Section~\ref{sec:async_result}. In Section~\ref{sec:bipartite}, we discuss the concept of structurally balanced networks and provide results for bipartite consensus. In Section \ref{sec:zero}, we present the network conditions for achieving zero consensus. Numerical examples that illustrate the theoretical findings are provided in Section~\ref{sec:simulation}. Finally, we conclude with a discussion in Section~\ref{sec:conclusion}.

\subsection{Notations}
\label{sec:notations}
We briefly present some useful notations used in this work. A zero matrix of dimension $a \times b$ is denoted by $\mathbf{0}_{a \times b}$ and a zero vector of dimension $d$ is denoted by $\mathbf{0}_d.$  A $d \times d$ identity matrix is denoted by $\mathbf{I}_d.$ A $d$-dimensional vector of ones is denoted by $\mathbf{1}_d.$ All the vectors considered here are column vectors unless otherwise specified. For simplicity of notations, we do not use the dimension subscript in these notations where ever it is clear. A symmetric matrix $S \in \mathbb{R}^{n \times n}$ is positive definite, if $y^TSy>0$ for all non-zero $y \in \mathbb{R}^n.$ The matrix $S$ is called negative definite if $-S$ is positive definite. A matrix-sign function $\text{sgn}(S) = 1$ if $S$ is positive definite. Similarly $\text{sgn}(S) = -1$ if $S$ is negative definite, and $\text{sgn}(S) = 0$ if $S = \mathbf{0}_{n \times n}.$ The $(i,j)$-th element of a matrix $S$ is denoted by $S^{(i,j)}.$ The infinite matrix norm for any matrix $S_{a\times b}$ is defined as the maximum absolute row sum of the matrix and is denoted by $||S|| = \max_{1\leq i \leq a} \sum_{j=1}^{b}|S^{(i,j)}|.$ A square matrix is said to be \textit{stochastic} if its every row is a probability vector.

\section{System Model}
\label{sec:model}

We represent a MAS by a matrix-weighted directed graph $G = (V, E, \mathcal{W})$ where $V$ is the set of $n$ agents and $E$ is the set of communication edges between them. An edge $(i,j) \in E$ if and only if an agent $i$ can receive information directly from an agent $j.$ The set of in-neighbors of an agent $i \in V$ is denoted by $\mathcal{N}_i = \{j \in V |(i, j) \in E\}.$ We assume that the graph $G$ has a directed spanning tree (also called rooted out-branching)\footnote{A graph is said to have a directed spanning tree if there exists an agent which has a directed path to every other agent.}. Each edge $(i,j) \in E$ has a weight associated with it which is denoted by a symmetric matrix $W_{ij} \in \mathbb{R}^{d \times d}.$ When $d = 1$ the weight is a scalar. The interaction weight matrix $\mathcal{W} = [W_{ij}] \in \mathbb{R}^{dn \times dn}$ is a block matrix such that the block identified by rows from $r_i=i+(i-1)(d-1)$ to $r_i+(d-1)$ and columns from $c_j=j+(j-1)(d-1)$ to $c_j+(d-1)$ is the matrix $W_{ij}.$ We assume that $W_{ii} = \mathbf{0}_{d \times d}, \forall i \in V.$ 

Each agent $i \in V$ maintains a $d$-dimensional state vector denoted by $X_i(t) \in \mathbb{R}^d$ at all times $t \in \mathbb{N}^+$ with $X_i(0)$ being the initial state vector. At each discrete time step $t \in \mathbb{N}^+,$ agents update their state vectors by linearly combining states of their neighbors. Let $\mathcal{X}(t):=\left[\mathcal{X}_1(t) ~ \mathcal{X}_2(t) ~ \dots ~ \mathcal{X}_d(t)\right]^T$ be a vector, where, $\mathcal{X}_k(t)$ be the $n$-dimensional vector consisting of $k^{th}$ dimension of the state vector of each agent at time $t.$ The consensus of state vectors is said to be achieved when all agents converge to same state vector. More formally,
\begin{definition}
	A network is said to achieve global consensus if $\lim\limits_{t \rightarrow \infty} \mid X_i(t)- X_j(t) \mid = \mathbf{0}_d,~ \forall i,j \in V.$ 
\end{definition}
The state updates can either happen \textit{synchronously} in which case all the agents update their states simultaneously or \textit{asynchronously} in which case only one agent updates its state at a time. We first present the synchronous update model and then modify it for the asynchronous case.

\subsection{Synchronous update model}
\label{sec:sync_model}

 In synchronous update model, all agents update their state vectors at the same time by weighted combination of their own and their neighbors' state vectors. Such update rule is commonly studied in the distributed scalar consensus literature \cite{olfati2007consensus,nedic2016convergence} and is based on the classical DeGroot model of opinion dynamics \cite{DeGroot74}. Generalization of this to vector states can be written as: $X_i(t+1) = (I_d -\tau \sum\limits_{j \in \mathcal{N}_i}\text{sgn}(W_{ij})W_{ij} X_i(t) + \tau \sum_{j \in \mathcal{N}_i}W_{ij} X_j(t).$ Here $\tau \in \mathcal{R}$ is called the step-size for the update. The range $\mathcal{R}$ is defined as $\mathcal{R} = \left(0,\frac{1}{  \max\limits_{i,k,m} 2\sum\limits_{j \in \mathcal{N}_i}\text{sgn}(W_{ij})W_{ij}^{(k,m)}} \right),$  where $W_{ij}^{(k,m)}$ is the $(k,m)$-th element of $W_{ij}$. As all agents update the states simultaneously, one can write this in matrix form as:
 
 \begin{equation}\label{eq:Update}
 	X(t)=\mathbf{P}X(t-1)=\mathbf{P}^tX(0).
 \end{equation} 
 Here $\mathbf{P}$ is of size   $nd \times nd.$ We rewrite this update rule by clubbing together each dimension of every agent's state vectors as follows for analysis. 
\begin{equation}\label{eq:SyncUpdate}
	\mathcal{X}(t)=\mathcal{F}\mathcal{X}(t-1)=(\mathcal{P}+\mathcal{Q})\mathcal{X}(t-1) = (\mathcal{P}+\mathcal{Q})^t \mathcal{X}(0),
\end{equation}
where, 
\begin{equation}\label{eq:P}
	\mathcal{P}=
	\begin{bmatrix}
		P_1 & \mathbf{0} &  \cdots & \mathbf{0}\\
		\mathbf{0} & P_2 &  \cdots & \mathbf{0}\\
		\vdots & \vdots & \ddots & \vdots \\
		\mathbf{0} & \mathbf{0} & \cdots & P_d
	\end{bmatrix}
\end{equation}
and
\begin{equation}\label{eq:Q}
	\mathcal{Q}=
	\begin{bmatrix}
		\mathbf{0} & Q_{12}&  \cdots & Q_{1d}\\
		Q_{21} & \mathbf{0} & \cdots & Q_{2d}\\
		\vdots & \vdots & \ddots & \vdots \\
		Q_{d1} & Q_{d2} & \cdots & \mathbf{0}
	\end{bmatrix}.
\end{equation}

Here $P_i \in \mathbb{R}^{n \times n}, Q_{ij} \in \mathbb{R}^{n \times n},$ and $\mathbf{0} \in  \mathbb{R}^{n \times n}.$ It is easy to verify that the $(k,m)$-th element of $P_i$ and $Q_{ij}$ matrices are as follows:

\begin{equation}
	P_i^{(k,m)} = \begin{cases}
		1- \tau \sum_{l \in \mathcal{N}_k} \text{sgn}(W_{kl}) W_{kl}^{(i,i)} ~~~ \text{for } k=m,\\
		\tau W_{km}^{(i,i)}\qquad \qquad \qquad \qquad ~~~~~ \text{for } k \neq m.
	\end{cases} \label{eq:Pi}
\end{equation}
\begin{equation}
	Q_{ij}^{(k,m)} = \begin{cases}
		- \tau \sum_{l \in \mathcal{N}_k} \text{sgn}(W_{kl}) W_{kl}^{(i,j)} ~~~ \text{for } k=m,\\
		\tau W_{km}^{(i,j)}\qquad \qquad \qquad \qquad   ~ \text{for } k \neq m.
	\end{cases} \label{eq:Qij}
\end{equation}

\begin{remark}\label{rm:P_stochastic}
	 Note that the structures of $P_i, Q_{ij}$ depend on the matrix sign function of the weight matrices $W_{km}.$
	 When the interaction weight matrix $W_{km}$ is positive definite, we have $\text{sgn}(W_{km}) = 1,$ and each diagonal element $W_{km}^{(i,i)} >0.$ By the definition of step-size, $\tau,$ it follows that every off-diagonal element of $P_i$ is also positive. Furthermore, it is straightforward to verify from \eqref{eq:Pi} that the entries in each row of $P_i$ sum to $1.$ Therefore, $P_i$ is a stochastic matrix with positive diagonal elements leading to matrix $\mathcal{P}$ also being a stochastic matrix with positive diagonal elements. Observe that the row sum of $Q_{ij}$ matrix is zero irrespective of the sign function of $W_{km}$ matrices.
\end{remark}

\subsection{Asynchronous update model}
\label{sec:asyn_model}

In asynchronous update model, only some agents update their states at a time; see \cite{Le-Phan24,Shi15,Boyd06} for some such models. In this work, we study the following asynchronous update model. At any time $t,$ an agent $i \in V$ is picked uniformly at random to update its state while all the other agents keep their states as it is. We assume that the agents are picked for update independently across time. If agent $l$ is picked at time $t$ for update then the update matrix $\mathbf{P}$ of \eqref{eq:Update} can be modified as:
\begin{equation}\label{eq:Update_random}
	X(t+1)=\mathbf{U}(t)X(t) = \prod_{k=0}^t \mathbf{U}(k) X(0)
\end{equation}
with $\mathbf{U}(t) = \mathbf{E}_l\mathbf{E}_l^T\mathbf{P}-\mathbf{E}_l\mathbf{E}_l^T+\mathbf{I} \in \mathbb{R}^{nd \times nd}.$ Here $\mathbf{E}_l \in \mathbb{R}^{nd \times d}$ is a block matrix contains $n$ blocks each of dimension $d \times d$. Note that the blocks are considered row-wise in $\mathbf{E}_l$. The $l^{th}$ block in $\mathbf{E}_l$ is an identity matrix of dimension $d$ and other blocks are zero matrices of dimension $d.$ This ensures that only agent $l$ updates its state and all the other agents' states remain unchanged. Similar to \eqref{eq:SyncUpdate}, one can write the update at time $t$ as:
 
\begin{equation} \label{eq:RandomUpdate1}
	\mathcal{X}(t+1)=\prod_{k=0}^t\mathcal{F}(k)\mathcal{X}(0) = \prod_{k=0}^t(\mathcal{P}(k)+\mathcal{Q}(k))\mathcal{X}(0).
\end{equation} 

\begin{remark}
	\label{rm:structure_P_Q}
Note that matrices $\mathbf{U}(t), \mathcal{P}(t), \mathcal{Q}(t)$ are random matrices as the agent to update its state is chosen randomly at any time $t.$ Structures of $\mathcal{P}(t)$ and $\mathcal{Q}(t)$ are the same as that of \eqref{eq:P} and \eqref{eq:Q}, respectively with the following difference: if an agent $l \in V$ is picked for an update at time $t,$ then the row corresponding to $l$ in $P_i,\forall i \in \{1,\ldots,d\}$ is non-zero with elements given in \eqref{eq:Pi}. All the other rows of $P_i$ have $1$ as their diagonal element and $0$ as off-diagonal elements. Similarly, the rows corresponding to agent $l$ in $Q_{ij}$ have non-zero elements given by \eqref{eq:Qij}, and all other rows are zero. Hence, even in asynchronous update model the random matrics $\mathcal{P}(t), \mathcal{Q}(t)$ have same structural properties as that of $\mathcal{P}, \mathcal{Q}$ matrices in synchronous model.
\end{remark}

Let the probability space generated by the random agent's picking be $(\Omega, \mathbb{F},\mathbb{P}).$ Here $\Omega = \{\omega: \omega=(v_1,v_2,\ldots), v_t \in V, t \in \mathbb{N}^+\},$ $\mathbb{F}$ is the $\sigma$-algebra generated by randomly picked agents, and $\mathbb{P}$ is the probability measure induced on the sample paths in $\Omega.$
\begin{remark}
	\label{rm:finite_choose}
	Recall that the number of agents in the network are $n < \infty.$ Thus each $\mathbf{U}(t)$ $(\text{in turn~} \mathcal{P}(t), \mathcal{Q}(t))$ can be one of the $n$ possible matrices each corresponding to an agent $i \in V$ being chosen at $t$ with probability of picking any agent as $\frac{1}{n}.$ In other words, in asynchronous model the update matrix at any time $t$ is chosen from a finite set with probability of choosing any element of the set being strictly positive.
\end{remark}

In this work, our aim is to analyze the state evolution in asynchronous update model. In the next section, we present some standard results from matrix convergence theory along with a proof of convergence for product of nonhomogenous matrices which are used to analyze the asynchronous model in Section~\ref{sec:async_result}.

\section{Matrix product convergence}
\label{sec:matrix_convergence}

Convergence of matrix products is a well studied area \cite{horn2012matrix, Hartfiel02} which is widely used in analyzing the state evolution in MAS. The following well known results from matrix convergence theory are used to prove our results in Section~\ref{sec:async_result}.

For any $n \times n$ matrix $S$ one can generate a graph $G_S = (V,E_S)$ in which an edge $(i,j) \in E_S$ if and only if $S^{(i,j)}>0.$ We call such a graph as the induced graph of matrix $S.$  
\begin{lemma}\cite{ren2005consensus}\label{lm:rank1convergence}
	Let $A \in \mathbb{R}^{n \times n}$ be a stochastic matrix with positive diagonal elements. Then, $\lim\limits_{t \rightarrow \infty} A^t$ converges to a rank one matrix with all identical rows if and only if the induced graph of $A$ has a spanning tree. 
\end{lemma}
\begin{lemma}\cite{horn2012matrix}\label{lm:zeroconvergence}
	Let $A \in \mathbb{R}^{n \times n}$ be a matrix with eigenvalue $\lambda.$ If $\mid \lambda \mid<1$ for all eigenvalues of $A,$ then $A$ is power convergent, i.e., $\lim\limits_{t \rightarrow \infty}A^t=\mathbf{0}.$ 
\end{lemma}
Now, we prove a result on the infinite product of matrices of the form $(A+B).$

\begin{lemma}
	\label{lm:cauchy}
	Let $C = A+B$ be a positive square matrix and $||.||$ be the infinite matrix norm. If $||A|| \leq 1,$ $\lim\limits_{t \rightarrow \infty} A^t$ exists, and $\lim\limits_{t \rightarrow \infty} B^t = \mathbf{0},$ then $\lim\limits_{t \rightarrow \infty} C^t \triangleq D$ exists, where $D$ is some positive matrix.
\end{lemma}

\begin{proof}
	We will prove that $\lim\limits_{t \rightarrow \infty} C^t$ exists by showing that the matrix sequence $J_t= C^t = (A+B)^t, \forall t \in \mathbb{N}^{+}$ is a Cauchy sequence. Let $J_{t,r}=A^{t-r}(A+B)^r$ for some $t > r,$ where $r \in \mathbb{N}^{+}$ is finite. By triangular inequality of the matrix norm, for some $t, s > r,$
	\begin{equation}
		||J_t-J_s|| \leq ||J_t-J_{t,r}||+||J_{t,r}-J_{s,r}||+||J_{s,r}-J_s||. \label{eq:J_expand}
	\end{equation}
	We will first show that for sufficiently large $t,s$ and $r,$ each term of the right side of \eqref{eq:J_expand} is less than $\epsilon/3$ for some constant $\epsilon>0.$ Expanding the first term and applying triangular inequality we get,
	\begin{align*}
		||J_t-J_{t,r}|| & = ||(A+B)^t-A^{t-r}(A+B)^r|| \\
		& \leq ||(A+B)^r||~||(A+B)^{t-r} - A^{t-r}||.
	\end{align*}
	For a fixed $r,$ $||(A+B)^r||$ is a constant. The second term here can be written as:
	\begin{align*}
		||(A+B)^{t-r} - A^{t-r}|| & \stackrel {\text{(a)}} {=} || A^{t-r-1}B + \ldots +B^{t-r} || \\
		& \stackrel {\text{(b)}} {\leq} || A^{t-r-1}B||+ \ldots +||AB^{t-r-1}|| \\&  +||B^{t-r}|| \\
		& \stackrel {\text{(c)}} {\leq} ||B|| + ||B^2||+ \ldots +||B^{t-r}|| \\
		& \stackrel {\text{(d)}} {\leq} \epsilon_1.
	\end{align*}
	Here the the equality (a) is by expanding the matrix term $(A+B)^{t-r},$ inequalities (b) and (c) are due to triangular inequality and the fact that $||A|| = 1$ respectively. For the inequality (d), note that there exist $\delta \in (0,1)$ such that $||B|| \leq \delta.$ As $\lim\limits_{t \rightarrow \infty} B^t \rightarrow \mathbf{0},$ there exist $t_1$ such that $||B^{t-r-s}-\mathbf{0}|| \leq \delta$ for all $t-r-s \geq t_1, ~\forall s=0,1,\ldots,(-r+t-1).$ By properly choosing $r,$ we get $||(A+B)^r||~\epsilon_1 \leq \epsilon/3.$ Which implies, $||J_t-J_{t,r}|| \leq \epsilon/3.$ Similarly, one can prove $||J_{s,r}-J_s|| \leq \epsilon/3.$ Now consider the second term on the right hand side of ~\eqref{eq:J_expand}.
	\begin{align*}
		||J_{t,r}-J_{s,r}|| & = ||A^{t-r}(A+B)^r-A^{s-r}(A+B)^r|| \\
		&=||(A+B)^r \left(A^{t-r}-A^{s-r}\right)||\\
		& \leq ||(A+B)^r||~||A^{t-r}-A^{s-r}||.
	\end{align*}
	For a fixed $r,$ $||(A+B)^r||$ is a constant. As $\lim\limits_{t \rightarrow \infty} A^t$ exists, $\exists ~t_2$ and $\epsilon_2>0$ such that $||A^{t-r}-A^{s-r}|| < \epsilon_2$ for all $t-r > t_2 \text{~and~} s-r>t_2.$ By properly choosing $r,$ we get $||(A+B)^r||~\epsilon_2 < \epsilon/3.$ Which implies, $||J_{t,r}-J_{s,r}|| < \epsilon/3.$ Hence using \eqref{eq:J_expand}, $||J_t-J_s|| < \epsilon$ which implies, $J_t$ is a Cauchy sequence. As the matrix product space with norm $||.||$ is complete, this shows that the Cauchy sequence $J_t$ converges thus proving the result.
\end{proof}

\begin{remark}
	Lemma~\ref{lm:cauchy} shows the convergence of infinite product of matrices of the form $(A+B)$ when $A$ and $B$ hold certain properties. Such matrix products are studied in nonhomogenous matrix theory which has applications in Markov chains and graphics; see \cite{Hartfiel02} for more details.  
\end{remark}

The following lemma shows that the induced graph of the product of matrices has more communication edges than the union of induced graphs of individual matrices.
\begin{lemma}(\cite{mounika2021opinion})\label{lm:async_induced}
	Let $A(1),A(2),\ldots,A(k)$ for some $k \geq 1$ be a sequence of stochastic matrices with positive diagonal elements. Then $\cup_{j=1}^{k}G_{A(j)} \subseteq G_{A(k)A(k-1)\ldots A(1)}.$
\end{lemma}
In other words, Lemma~\ref{lm:async_induced} says that the union of induced graphs of individual matrices is a subgraph of the induced graph of product of these matrices. This is used to prove the following result about existence of a spanning tree in the finite union of the induced graphs of update matrices. 

\begin{lemma}\label{lm:async_sp}
	Let $G = (V, E, \mathcal{W})$ be a matrix-weighted network with a spanning tree and positive definite edge weights. If agents update their states using \eqref{eq:RandomUpdate1}, then the induced graph of $H_i=\lim \limits_{t \rightarrow \infty} \prod_{k=0}^tP_i(k), \forall i \in \{1,2,\ldots,d\}$ has a spanning tree almost surely.
\end{lemma}
\begin{proof}
We fix an index $i \in \{1,\ldots,d\}$ and prove the result for it. The structure of each $P_i(k)$ matrix depends on the agent picked for update at time $k$ and is same for all $i.$ Recall that the matrix $P_i(k),$ from \eqref{eq:Pi} and Remark~\ref{rm:structure_P_Q}, is a stochastic matrix with positive diagonal elements. Since stochastic matrices with positive diagonal elements are closed under multiplication \cite{horn2012matrix}, $\prod_{k=0}^t P_i(k)$ is a stochastic matrix with positive diagonal elements. Recall that every agent is chosen at any time with positive probability (Remark \ref{rm:finite_choose}). Let $Y_l=P_i(1+(l-1)n)P_i(2+(l-1)n)\ldots P_i(ln)$ for all $l \in \mathbb{N}^{+}$ be a random matrix obtained by multiplying $n,$ $P_i(k)$ matrices. It is easy to observe that the induced graph $G_{Y_l}$ of $Y_l$ depends only of the sequence of agents picked for updates and not on the the index $i.$ According to Lemma \ref{lm:async_induced}, the induced graph $G_{Y_l}$ of $Y_l$ contains more communication edges than the union of induced graphs $\cup_{k=1}^{n}G_{P_i(k+(l-1)n)}.$  The probability that a spanning tree $\mathcal{T}$ is included in the induced graph of $Y_l$ for any $l$ is: $\mathbb{P}(\mathcal{T} \subseteq G_{Y_l}) \geq \mathbb{P}(\mathcal{T} \subseteq \cup_{k=1}^{n}G_{P_i(k+(l-1)n)}) \geq \left( \frac{1}{n} \right)^n.$ Here the last inequality holds by the fact that at every time an agent is selected uniformly at random for update with each selection being independent of previous ones. Let $E_l$ be the event that there is a spanning tree in the induced graph of $Y_l.$ Thus, $\mathbb{P}(E_l) \geq \left( \frac{1}{n} \right)^n>0.$ Note that $\sum_{l=1}^{\infty}\mathbb{P}(E_l) = \infty$ and the events $\left( E_l \right)_{l=1}^{\infty}$ are independent. Hence, by the second Borel-Cantelli lemma, the event $E_l$ occurs infinitely often almost surely, i.e., $\mathbb{P}(E_l \text{~i.o.~})=1.$ Thus, the induced graph of $\lim\limits_{t \rightarrow \infty} \prod_{k=0}^tP_i(k), ~\forall i \in \{1,\ldots,d\}$ has a spanning tree almost surely. 
\end{proof}
\begin{remark}
	\label{rm:samplepath}
	Let $\Omega' \subseteq \Omega$ be the set of sample paths for which Lemma~\ref{lm:async_sp} holds. Then, by Lemma~\ref{lm:async_sp}, $\mathbb{P}(\Omega') =1.$ For any sample path $\omega \in \Omega',$ let $t'(\omega) < \infty$ be the time such that $E_l$ occurs infinitely often almost surely for all $t \geq t'(\omega).$  Implying the induced graph of $\prod_{k=0}^tP_i(k),~\forall i$ has a spanning tree almost surely for all $t \geq t'(\omega).$
\end{remark}

In the next section we use these results to prove convergence of global consensus.
%

\section{Global Consensus in asynchronous model}
\label{sec:async_result}

In this section, we analyze the evolution of state vectors for update rule \eqref{eq:RandomUpdate1} when every edge weight matrix $W_{ij}$ is positive definite in the graph $G.$ To prove the convergence of state vectors, we need to prove the convergence of $\prod_{k=0}^t\mathcal{F}(k)$ which in turn requires the convergence of  $\lim\limits_{t \rightarrow \infty} \prod_{k=0}^t\mathcal{P}(k)$ and $\lim\limits_{t \rightarrow \infty} \prod_{k=0}^t\mathcal{Q}(k)$. Recall that the matrix $\mathcal{F}(k)$ is dictated by the sample path $\omega$ coming from the sample space $\Omega$ and thus can be denoted as $\mathcal{F}_\omega(k).$ We prove the convergence of product of $\mathcal{F}_\omega(k)$ for every sample path $\omega \in \Omega'$ (recall definition of $\Omega'$ from Remark~\ref{rm:samplepath}). For simplicity, we skip the subscript $\omega$ from the matrix notation whenever it clear from the context.
 
We first show that $\prod_{k=0}^t\mathcal{P}(k)$ converges to a block diagonal matrix with a rank-1 matrix in each block. The following result is generalization of a standard result (see Lemma $3.7$ of \cite{ren2005consensus}) on the convergence of product of stochastic matrices, which can be proved by combining Lemma~\ref{lm:async_sp} and property of $\mathcal{P}(k)$ matrices as mentioned in Remark~\ref{rm:structure_P_Q}.
\begin{lemma}
	\label{lm:random_P}
	Let $G = (V, E, \mathcal{W})$ be a matrix-weighted network with a spanning tree and positive definite edge weights. If agents update their states using \eqref{eq:RandomUpdate1}, then for every sample path $\omega \in \Omega',$ $\lim\limits_{t \rightarrow \infty} \prod_{k=0}^t\mathcal{P}(k)$ converges to a block diagonal matrix with a rank-1 matrix in each block of the diagonal.
\end{lemma}
\begin{proof}
	On any sample path $\omega \in \Omega',$ by the definition of matrices $\mathcal{P}(k)$ (see Remark~\ref{rm:structure_P_Q} and\eqref{eq:P}) observe that $\prod_{k=0}^t\mathcal{P}(k)$ is a block diagonal matrix with each block $i \in \{1,\ldots,d\}$ being $\prod_{k=0}^t P_i(k).$ We skip the notation $\omega$ for simplicity. By Remark~\ref{rm:structure_P_Q}, for any $i \in \{1,\ldots,d\}$ and for any $k \in \mathbb{N}^+,$ $P_i(k)$ matrix is stochastic with positive diagonal elements thus $\prod_{k=0}^t P_i(k)$ is also a stochastic matrix with positive diagonal elements. By Lemma~\ref{lm:async_sp}, the induced graph of $\lim\limits_{t \rightarrow \infty} \prod_{k=0}^t P_i(k)$ has a spanning tree on the sample path $\omega.$ Thus, by Lemma~3.9 of \cite{ren2005consensus}, $\lim\limits_{t \rightarrow \infty} \prod_{k=0}^t P_i(k)$ converges to a rank-1 matrix. This proves the result.  
\end{proof}

Now, we show that the product $\prod_{k=0}^t\mathcal{Q}(k)$ converges to a zero matrix for a sample path $\omega \in \Omega'$ whose proof is given in Appendix \ref{sec:app_a}.
\begin{lemma}
	\label{lm:random_Q}
	Let $G = (V, E, \mathcal{W})$ be a matrix-weighted network with a spanning tree and positive definite edge weights. If agents update their states using \eqref{eq:RandomUpdate1}, then for every sample path $\omega \in \Omega',\lim\limits_{t \rightarrow \infty} \prod_{k=0}^t\mathcal{Q}(k)$ converges almost surely to a zero matrix.
\end{lemma}

By using properties proved in Lemmas~\ref{lm:random_P} and \ref{lm:random_Q} we now show the existence of convergence of $\prod_{k=0}^t\mathcal{F}(k)$ on the lines of proof of Lemma~\ref{lm:cauchy}. 

\begin{lemma}\label{lm:prod_F(t) convergence}
	Let $G = (V, E, \mathcal{W})$ be a matrix-weighted network with a spanning tree and positive definite edge weights. If agents update their states using \eqref{eq:RandomUpdate1}, then for every sample path $\omega \in \Omega',$ the matrix sequence $\prod_{k=0}^t\mathcal{F}(k), \forall t \geq t'(\omega)$ is a Cauchy sequence.
\end{lemma}

Formal proof is given in Appendix \ref{sec:app_b}. Note that the above result shows that $\prod_{k=0}^t\mathcal{F}(k)$ converges to a finite matrix $\mathcal{B}(\omega)$ on sample path $\omega$. In the following results, we show some properties of eigenvalues of $\prod_{k=0}^t\mathcal{F}(k), \forall t \geq t'(\omega)$ which will be used to prove the structure of $\mathcal{B}(\omega).$

\begin{lemma} \label{lm:multiplicity_eigen_random}
	Let $G = (V, E, \mathcal{W})$ be a matrix-weighted network with a spanning tree and positive definite edge weights. If agents update their states using \eqref{eq:RandomUpdate1}, then for every sample path $\omega \in \Omega',$ $\prod_{k=0}^t\mathcal{F}(k), \forall t \geq t'(\omega)$ has an eigenvalue $\lambda=1$ with algebraic multiplicity $d$ and geometric multiplicity $d.$
\end{lemma}

See Appendix \ref{sec:app_c} for proof of Lemma \ref{lm:multiplicity_eigen_random}. Now, we show that no eigenvalue $\lambda$ of $\prod_{k=0}^t\mathcal{F}(k), \forall t \geq t'(\omega)$ has $|\lambda| > 1.$
%
\begin{lemma} \label{lm:PQ_eigen_less_random}
	Let $G = (V, E, \mathcal{W})$ be a matrix-weighted network with a spanning tree and positive definite edge weights. If agents update their states using \eqref{eq:RandomUpdate1}, then for every sample path $\omega \in \Omega',$ no eigenvalue $\lambda$ of $\prod_{k=0}^t\mathcal{F}(k), \forall t \geq t'(\omega)$ has $|\lambda| > 1.$
\end{lemma}
\begin{proof}
	We prove the result by contradiction. Let there be an eigenvalue $\lambda_1$ of $\prod_{k=0}^t\mathcal{F}(k), \forall t \geq t'(\omega)$ such that $|\lambda_1|>1.$ By Lemma~\ref{lm:prod_F(t) convergence}, we know that $\lim\limits_{t \rightarrow \infty} \prod_{k=0}^t\mathcal{F}(k)$ exists. For every $t \geq t'(\omega)$, the Jordan canonical form of $\prod_{k=0}^t\mathcal{F}(k)$ can be written as $\prod_{k=0}^t\mathcal{F}(k)=S\prod_{k=0}^tJ(k)S^{-1}$ where $S$ is a matrix with eigenvectors and $\prod_{k=0}^tJ(k)$ is the Jordan matrix of eigenvalues of $\prod_{k=0}^t\mathcal{F}(k).$ It is easy to verify that $\lim\limits_{t \rightarrow \infty} \prod_{k=0}^t\mathcal{F}(k) = S \lim\limits_{t \rightarrow \infty}\prod_{k=0}^tJ(k)S^{-1}.$ As each element of $S$ is finite, by Lemma~\ref{lm:prod_F(t) convergence}, $\lim\limits_{t \rightarrow \infty}\prod_{k=0}^tJ(k)$ exists. By Lemma~\ref{lm:multiplicity_eigen_random}, we know that $\prod_{k=0}^t\mathcal{F}(k)$ has $d$ unity eigenvalues and $(nd-d)$ non-unity eigenvalues. The Jordan matrix of $\prod_{k=0}^t\mathcal{F}(k)$ can be written as 
	$\prod_{k=0}^tJ(k) =\begin{bmatrix}
		\textbf{J}_{\lambda \neq 1} & \textbf{0}_{(nd-d) \times d}\\ 
		\textbf{0}_{d \times (nd-d)} &  \textbf{I}_{d}\\
	\end{bmatrix}.$
	It has $d$ Jordan blocks of size $1 \times 1$ corresponding to unity eigenvalue whose algebraic and  geometric multiplicities are $d.$ The Jordan block corresponding to non-unity eigenvalues is $ \textbf{J}_{\lambda \neq 1}$ which has non-unity eigenvalues in its diagonal and zeros on off-diagonal elements. Observe that $\lim\limits_{t \rightarrow \infty} \textbf{J}_{\lambda \neq 1}^t$ will diverge to infinity because $|\lambda_1| >1.$ This implies that $\lim\limits_{t \rightarrow \infty}  \prod_{k=0}^tJ(k)$ does not exists which contradicts the result of Lemma~\ref{lm:prod_F(t) convergence}. Recall by Lemma \ref{lm:async_sp} that the event $E_l$ (existence of spanning tree in finite time steps) happens infinitely often. Hence the argument for $\textbf{J}_{\lambda \neq 1}$ holds true for every consecutive $t'(\omega)<\infty$ time slots. Hence, the matrix $\prod_{k=0}^t\mathcal{F}(k), \forall t \geq t'(\omega)$ can not have any eigenvalue greater than one.
\end{proof}
Now we present the main result of this work.

\begin{theorem}\label{th:random}
	Let $G = (V, E, \mathcal{W})$ be a matrix-weighted network with a spanning tree and positive definite edge weights. If agents update their states using \eqref{eq:Update_random}, then the network achieves global consensus asymptotically almost surely for all initial states.
\end{theorem}
\begin{proof}
	For any sample path $\omega \in \Omega',$ by Lemma~\ref{lm:multiplicity_eigen_random} and Lemma~\ref{lm:PQ_eigen_less_random} we can write:
	$\lim\limits_{t \rightarrow \infty}\prod_{k=0}^tJ(k)=\begin{bmatrix}
		\textbf{0}_{(nd-d) \times (nd-d)} & \textbf{0}_{(nd-d) \times d}\\ 
		\textbf{0}_{d \times (nd-d)} &  \textbf{I}_{d}\\
	\end{bmatrix}.$
	Hence, the state update rule~\eqref{eq:RandomUpdate1} can be written as:
	\begin{equation*}
		\lim_{t \rightarrow \infty} \mathcal{X}(t) 
		=S \begin{bmatrix}
			\textbf{0}_{(nd-d) \times (nd-d)} & \textbf{0}_{(nd-d) \times d}\\ 
			\textbf{0}_{d \times (nd-d)} &  \textbf{I}_{d}\\
		\end{bmatrix} S^{-1} \mathcal{X}(0).
	\end{equation*}
	The last $d$ rows of $S^{-1}$ remain same and all other rows become zero in the product $\lim\limits_{t \rightarrow \infty}\prod_{k=0}^tJ(k) S^{-1}.$ The last $d$ columns of $S$ are linearly independent eigenvectors corresponding to unity eigenvalue of $\prod_{k=0}^t\mathcal{F}(k).$ See the definition of these vectors $\overrightarrow{V_1},\cdots,\overrightarrow{V_d}$ in the proof of Lemma~\ref{lm:multiplicity_eigen_random}. Without loss of generality, let $k=1$ in $\overrightarrow{V_1},\cdots,\overrightarrow{V_d}.$ Then each non-zero row of $\lim\limits_{t \rightarrow \infty}\prod_{k=0}^tJ(k) S^{-1}$ repeats $n$ times when we multiply $S$ with $\lim\limits_{t \rightarrow \infty}\prod_{k=0}^tJ(k) S^{-1}.$ Thus, $S \lim\limits_{t \rightarrow \infty}\prod_{k=0}^tJ(k) S^{-1}$ converges to a matrix $\mathcal{B}(\omega).$ Matrix $\mathcal{B}(\omega)$ contains $d$ horizontal blocks, i.e., $\mathcal{B}(\omega)=[R_1 ~R_2 ~ \dots ~R_d]^T$, where each block $R_i \in \mathbb{R}^{n \times nd}$ is a rank-1 matrix. Thus, $k^{th}$ dimension of all agents converge to the same value on sample path $\omega$. As $\omega \in \Omega'$ and $\mathbb{P}(\Omega')=1.$ This proves the almost sure convergence of global consensus.
\end{proof}

Observe that all the results for the asynchronous system model are proved on sample paths for which the existence of spanning tree on the induced graph of update matrices in finite time is guaranteed by Lemma~\ref{lm:async_sp}. Recall that the underlying graph $G$ has a spanning tree and thus in synchronous time model the spanning tree exists for the update matrix $\mathbf{P}$ (refer \eqref{eq:Update}) at all times. Thus all results of this section hold true for the synchronous time model as well which is formally presented here.

\begin{corollary}
	\label{cr:synch}
	Let $G = (V, E, \mathcal{W})$ be a matrix-weighted network with a spanning tree and positive definite edge weights. If agents synchronously update their states by \eqref{eq:Update}, then the network achieves global consensus asymptotically almost surely for all initial states.
\end{corollary}

\begin{example}\label{ex:tenode_globalconsensus}
	See a network of 10 agents and 20 edges containing a spanning tree in Figure~\ref{fig:Tennodenw}. Each edge is assigned a positive definite weight matrix. Additionally, the initial state of each agent is defined by a randomly selected three-dimensional vector. Figure~\ref{fig:global1} illustrates the evolution of the state vectors leading to global consensus when states are updated using \eqref{eq:Update_random}. Note that the final consensus vector depends on the initial state vectors and the sequence of agents selected for update.
\end{example}
\begin{figure}
	\center \includegraphics[width=0.5\linewidth]{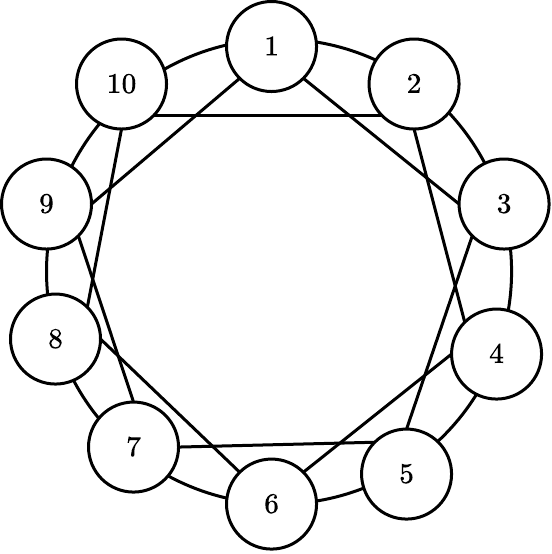}
	\caption{A $10$-agents $4$-regular network with a spanning tree.}
	\label{fig:Tennodenw}
\end{figure}
\begin{figure*}
	\centering
	\subfigure[]{%
		\includegraphics[width=0.3\linewidth]{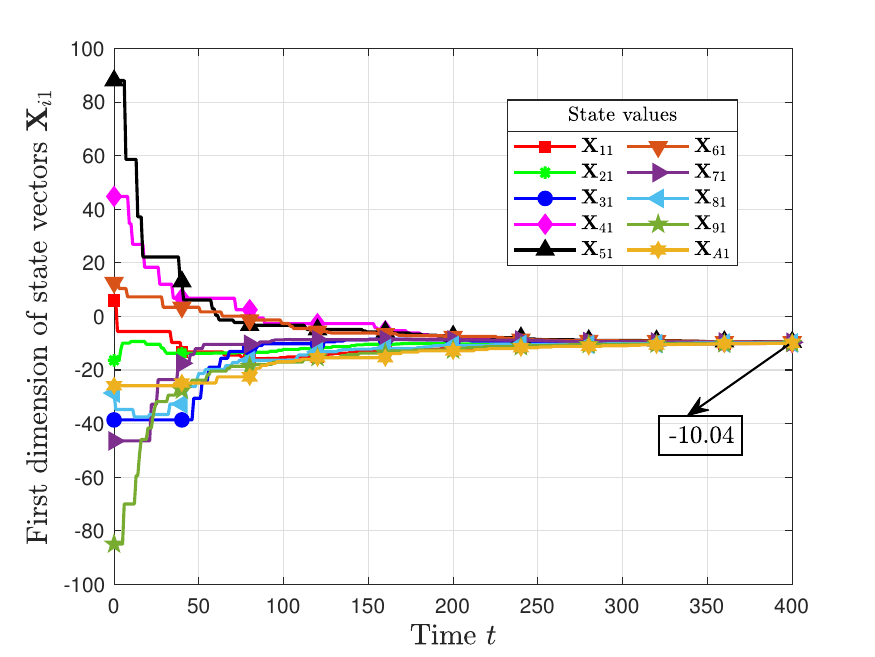}
		\label{fig:globalconsensus1}}
	\quad
	\subfigure[]{%
		\includegraphics[width=0.3\linewidth]{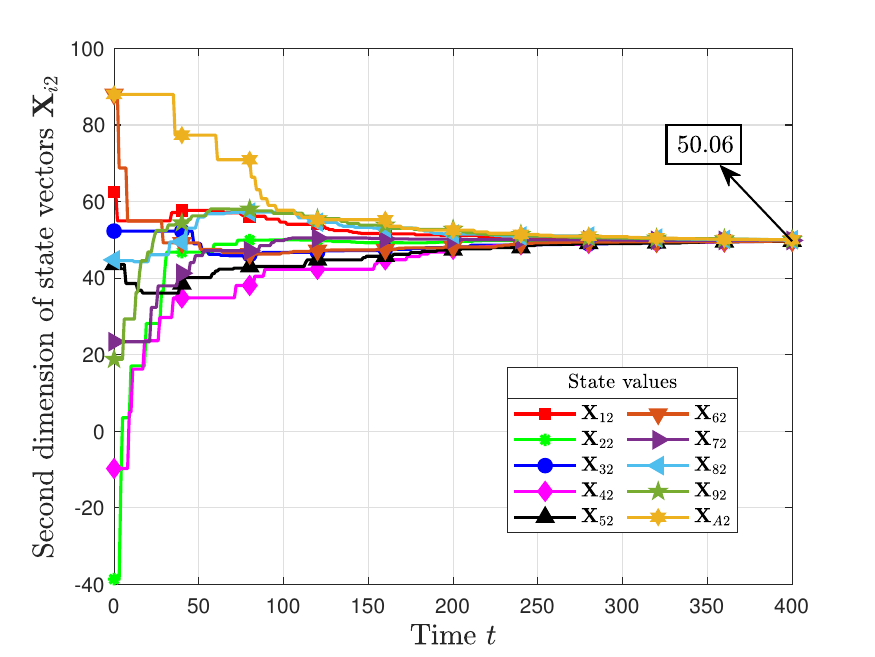}
		\label{fig:globalconsensus2}}
	\quad
	\subfigure[]{%
		\includegraphics[width=0.3\linewidth]{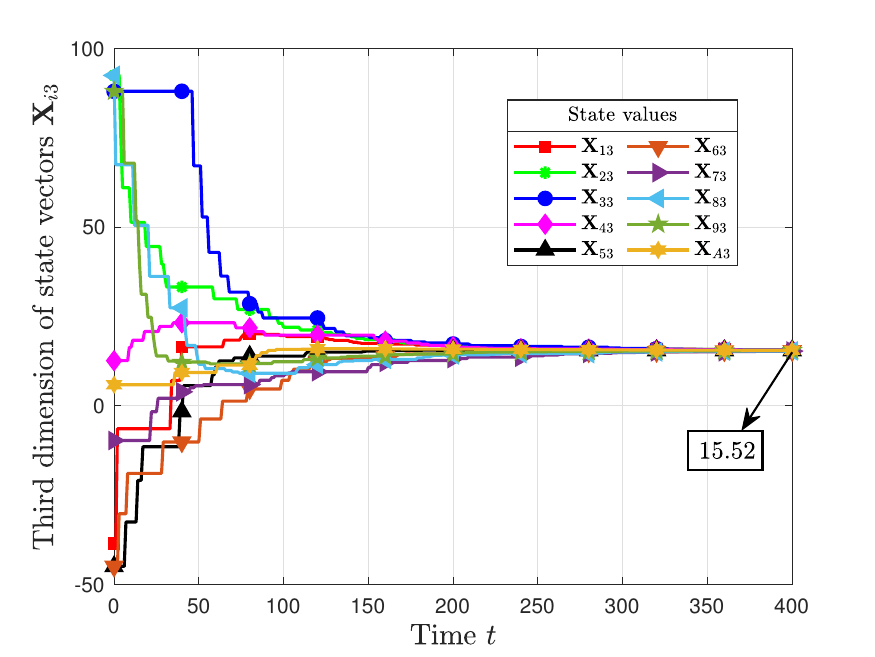}
		\label{fig:globalconsensus3}}
	\caption{ Evolution of state vectors of an agent $i \in V$ for the network shown in Figure \ref{fig:Tennodenw} with asynchronous update rule \eqref{eq:Update_random} showing global consensus. (a) First dimension of state vectors. (b) Second dimension of state vectors. (c) Third dimension of state vectors.}
	\label{fig:global1}
\end{figure*}

Note that the convergence of $\mathcal{P}(k)$ and $\mathcal{Q}(k)$ (shown in Lemma \ref{lm:random_P} and \ref{lm:random_Q}) depend on the fact that all $W_{ij}$ matrices are positive definite. In the next section, we analyze the state evolution when some edge weights are negative definite.
%

\section{Bipartite consensus in synchronous model}
\label{sec:bipartite}

In Section~\ref{sec:async_result}, we showed that the global consensus is achieved in network $G$ if it has a spanning tree (both asynchronous and synchronous time model) and all edge weights are positive definite. In this section we discuss the state evolution when some edge weights are negative definite inspired by the signed graphs in scalar consensus literature \cite{Liu17}. A negative definite weight matrix for an edge can represent competition (or mistrust) among the agents while positive definite weight can represent cooperation. In such a network, we observe that the network does not achieve global consensus instead the type of the consensus depends on the structure of the network. 
\begin{definition}\cite{Su20}
	A network is said to achieve a bipartite consensus if the agent set $V$ can be partitioned into two sets $V_1, V_2$ such that $\lim\limits_{t \rightarrow \infty}{X_i=\mathbf{C}}, ~\forall i \in V_1,$ and $\lim\limits_{t \rightarrow \infty}{X_j=-\mathbf{C}}, ~\forall j \in V_2,$ where $\mathbf{C}$ is a non-zero vector of dimension $d.$  A network is said to achieve zero consensus if all the state vectors converge to zero vector, i.e., $\lim\limits_{t \rightarrow \infty} X_i(t) =\mathbf{0}_d, ~\forall i \in V.$
\end{definition}

In scalar consensus, bipartite consensus is achieved if the network is structurally balanced. The generalization of structurally balanced property on matrix-weighted networks is presented here.

\begin{definition}\cite{Pan19}\label{def:balanced}
	A matrix-weighted network $G=(V,E,\mathcal{W})$ is called structurally balanced if there exists a bi-partition of agents $V = V_1 \cup V_2$ such that $W_{ij}$ is positive definite if $i,j \in V_k, ~\forall k \in \{1,2\},$ and $W_{ij}$ is negative definite if $i \in V_k, j \in V \setminus V_k.$
\end{definition}

See example of a structurally balanced network in Figure \ref{fig:SB}.

A network with positive and negative definite edge weights but is not structurally balanced is called \textit{structurally unbalanced} network. 
For simplicity of exposition, in this work we present the analysis for the synchronous update model but the analysis holds true for asynchronous update model as well. Numerical example for asynchronous update model is presented in Example \ref{ex:tenode_bipartiteconsensus}. Now we present and prove that bipartite consensus is achieved on structurally balanced networks.
\begin{theorem}
	\label{th:bipartite}
	Let $G = (V, E, \mathcal{W})$ be a matrix-weighted structurally balanced network with a spanning tree. If agents synchronously update their state by \eqref{eq:Update}, then the network achieves bipartite consensus asymptotically almost surely for all initial states.
\end{theorem}
\begin{proof} 
Let $V_1$ and $V_2$ be the partition of agents following structurally balanced condition of Definition~\ref{def:balanced}.
Without loss of generality, we label agents such that first $k$ agents belong to $V_1$ and the remaining $n-k$ agents belong to $V_2.$ Let $\Delta \coloneqq \text{diag}\{\mathbf{I}_d,\ldots,\mathbf{I}_d,-\mathbf{I}_d,\ldots,-\mathbf{I}_d\}$ be a block diagonal matrix with the matrices $\mathbf{I}, -\mathbf{I}$ repeating $k$ and $n-k$ times respectively. Let $\mathbf{D} \coloneqq \Delta^{-1}\mathbf{P}\Delta$ and a vector $Y(t) \coloneqq \Delta^{-1}X(t).$ Then, from update rule \eqref{eq:Update} we can write the update rule for vector $Y$ as:
\begin{equation}\label{eq:Yupdate}
	Y(t+1)=\mathbf{D}Y(t).
\end{equation}
By rearranging the vector $Y$ to get $k$-th dimension of the state vectors together, \eqref{eq:Yupdate} can be written as:
\begin{equation}\label{eq:Yupdate1}
	\mathcal{Y}(t)=(\mathcal{S+T})\mathcal{Y}(t-1)=(\mathcal{S+T})^t\mathcal{Y}(0),
\end{equation}
where, 
\[ \mathcal{S}=
	\begin{bmatrix}
		S_1 & \mathbf{0} &  \cdots & \mathbf{0}\\
		\mathbf{0} & S_2 &  \cdots & \mathbf{0}\\
		\vdots & \vdots & \ddots & \vdots \\
		\mathbf{0} & \mathbf{0} & \cdots & S_d
	\end{bmatrix},
	\mathcal{T}=
	\begin{bmatrix}
		\mathbf{0} & T_{12}&  \cdots & T_{1d}\\
		T_{21} & \mathbf{0} & \cdots & T_{2d}\\
		\vdots & \vdots & \ddots & \vdots \\
		T_{d1} & T_{d2} & \cdots & \mathbf{0}
	\end{bmatrix}
\]
with $S_i \in \mathbb{R}^{n \times n}, T_{ij} \in \mathbb{R}^{n \times n},$ and $\mathbf{0} \in \mathbb{R}^{n \times n}.$ The matrices $S_i$ and $T_{ij}$ can be obtained by rearranging the elements of matrix $\mathbf{D}$ with their $(k,m)$-th elements as:
\begin{equation*}
	S_i^{(k,m)}= \begin{cases}
		1-\tau \sum\limits_{l \in \mathcal{N}_k^T}{W_{kl}^{(i,i)}}+\tau \sum\limits_{l \in \mathcal{N}_k^{NT}}{W_{kl}^{(i,i)}}~\text{for~} k=m,\\
		\tau W_{km}^{(i,i)}~~~~~~~~~~~~~~~~~~~\text{for~} k \neq m ~\text{\&}~ m \in \mathcal{N}_k^T,\\
		-\tau W_{km}^{(i,i)}~~~~~~~~~~~~~~~~\text{for~} k \neq m ~\text{\&}~ m \in \mathcal{N}_k^{NT}.
	\end{cases}
\end{equation*}
\begin{equation*}
	T_{ij}^{(k,m)}= \begin{cases}
		-\tau \sum\limits_{l \in \mathcal{N}_k^T}{W_{kl}^{(i,j)}}+\tau \sum\limits_{l \in \mathcal{N}_k^{NT}}{W_{kl}^{(i,j)}}~~~~~\text{for~} k=m,\\
		\tau W_{km}^{(i,j)}~~~~~~~~~~~~~~~~~~~\text{for~} k \neq m  ~\text{\&}~ m \in \mathcal{N}_k^T,\\
		-\tau W_{km}^{(i,j)}~~~~~~~~~~~~~~~~\text{for~} k \neq m  ~\text{\&}~ m \in \mathcal{N}_k^{NT}.
	\end{cases}
\end{equation*}
Here $\mathcal{N}_k^T$ is the set of in-neighbors of an agent $k \in V_i$ from the same set $V_i$ and $\mathcal{N}_k^{NT}$ is the set of in-neighbors of an agent $k$ from $V \setminus V_i.$ In other words, edge matrix $W_{km}$ is positive definite if $m \in \mathcal{N}_k^T$ and is negative definite if $m \in \mathcal{N}_k^{NT}.$ It is easy to verify that the $S_i$ and $T_{ij}$ have the same properties as that of $P_i$ and $Q_{ij}$ (see \eqref{eq:Pi}, \eqref{eq:Qij}) respectively. Thus the almost sure convergence of global consensus for $Y$ in \eqref{eq:Yupdate} can be proved on the similar lines of $X$ in \eqref{eq:Update} with the arguments presented in Section~\ref{sec:async_result}. By definition of $Y$ and $\mathbf{D},$ observe that if the vector $Y_i$ of an agent $i$ converges to $\mathbf{C},$ then $X_i$ converges to $\mathbf{C}$ for $i \in V_1$ and to $-\mathbf{C}$ for $i \in V_2.$ This proves that the network achieves bipartite consensus when the network is structurally balanced.
\end{proof}

\begin{example} \label{ex:tenode_bipartiteconsensus}
	A structurally balanced network of $10$ agents with a spanning tree is shown in Figure~\ref{fig:SB}. The edge weight matrices are either positive definite or negative definite which are chosen randomly. Here agent set partitions as $V_1=\{1,3,5,7,9\}$ and $V_2 =\{2,4,6,8,10\}.$ State vectors evolution when agents update their states using \eqref{eq:Update_random} is shown in Figure \ref{fig:bipartite}. Observe that all the agents in $V_1$ converge to a state vector $C$ whereas all agents in $V_2$ converge to $-C$ thus achieving bipartite consensus.
\end{example}
\begin{figure}
	\center \includegraphics[width=0.5\linewidth]{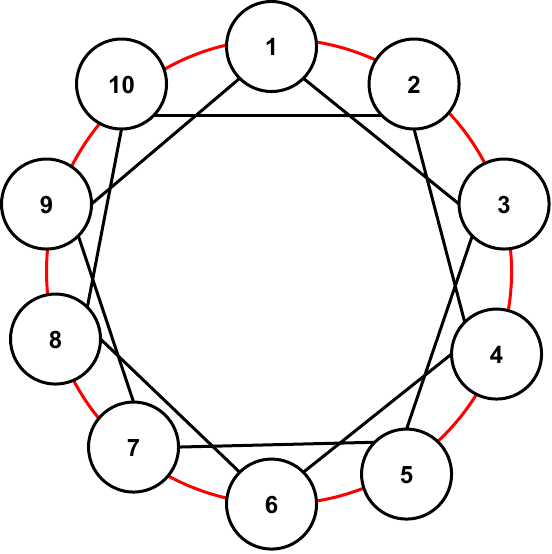}
	\caption{A structurally balanced network with $V_1 =\{1,3,5,7,9\}$ and $V_2 =\{2,4,6,8,10\}.$ Here black and \textcolor{red}{red} edges represent positive definite and \textcolor{red}{negative definite} weight matrices respectively.}
	\label{fig:SB}
\end{figure}
\begin{figure*}
	\centering
	\subfigure[]{%
		\includegraphics[width=0.3\linewidth]{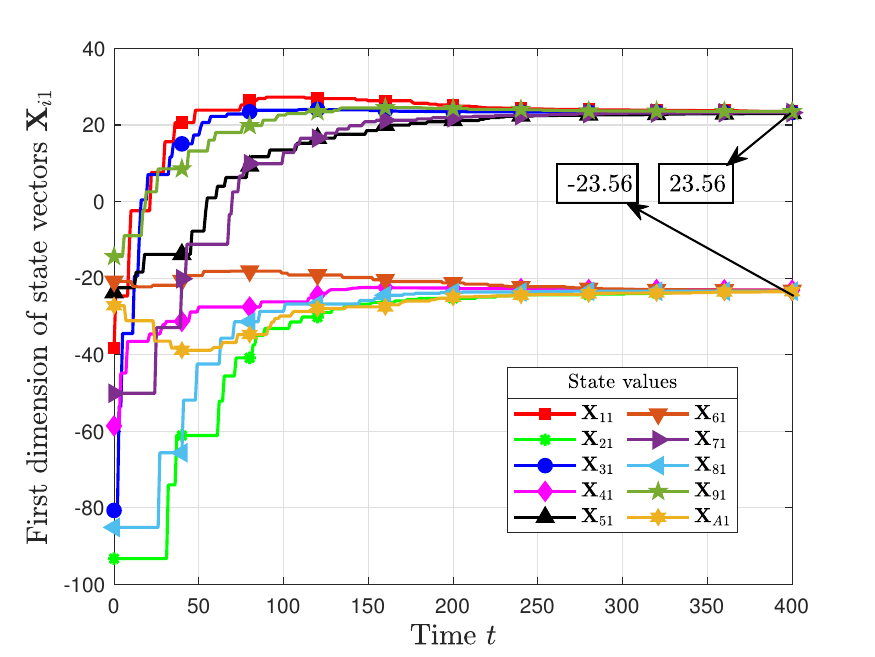}
		\label{fig:BC1}}
	\quad
	\subfigure[]{%
		\includegraphics[width=0.3\linewidth]{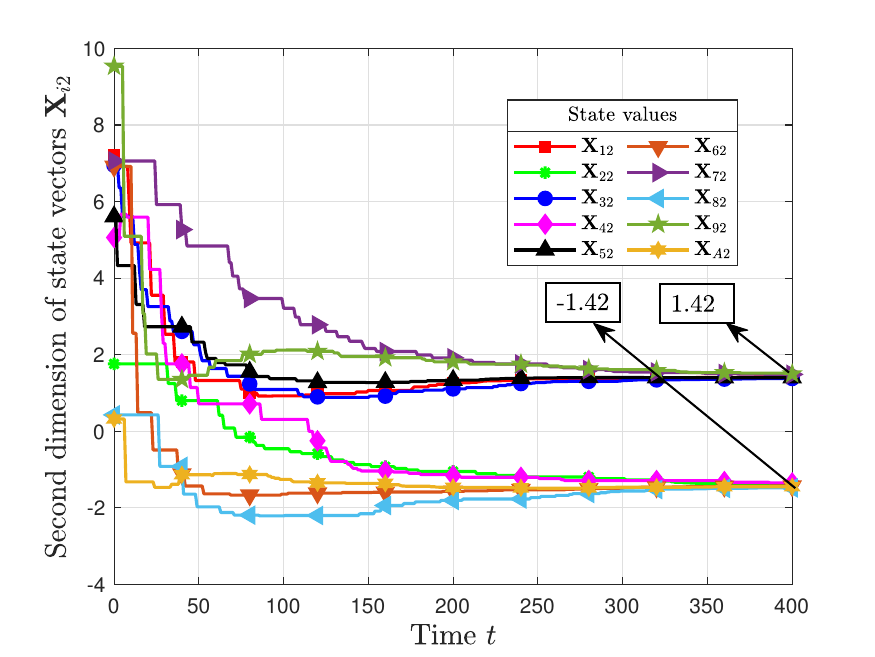}
		\label{fig:BC2}}
	\quad
	\subfigure[]{%
		\includegraphics[width=0.3\linewidth]{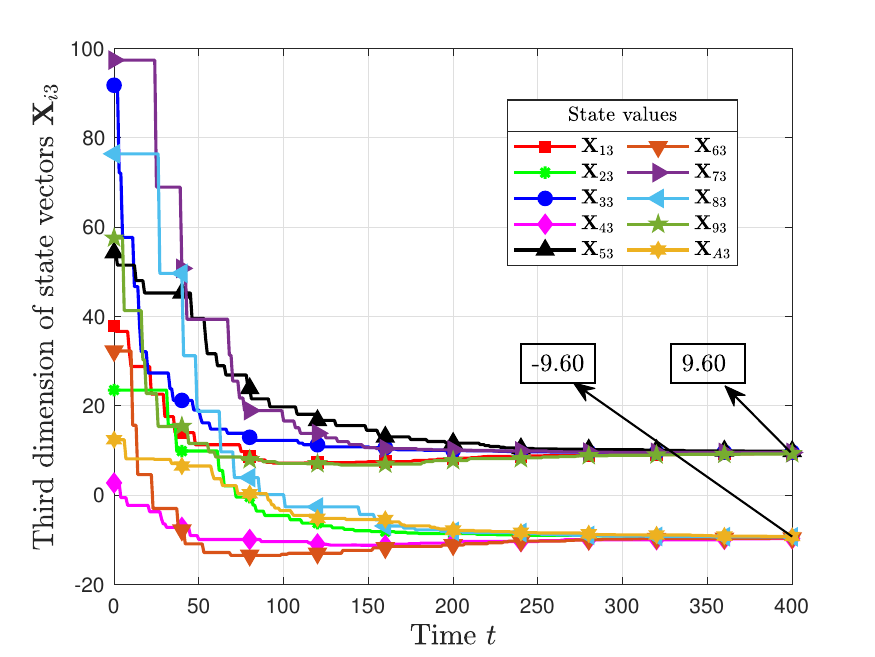}
		\label{fig:BC3}}
	\caption{ Evolution of state vectors of an agent $i \in V$ for the network shown in Figure \ref{fig:SB} with asynchronous update rule \eqref{eq:Update_random} showing bipartite consensus. (a) First dimension of state vectors. (b) Second dimension of state vectors. (c) Third dimension of state vectors.}
	\label{fig:bipartite}
\end{figure*}

In the next section, we study network condition for achieving zero consensus.
%

\section{Zero consensus in synchronous model}
\label{sec:zero}

 In this section, we show that under certain network conditions the agents achieve zero consensus. In particular, in Section~\ref{sec:negative} we study networks with all negative definite edge weights and in Section~\ref{sec:unbalanced} we study structurally unbalanced networks.   

\subsection{All edge weights are negative definite}
\label{sec:negative}

In this section, we consider a network in which every edge weight matrix is negative definite and prove that all the agents' states reach to zero vector when agents update their states synchronously \eqref{eq:Update}. Our results also hold true for asynchronous update model. Numerical example for asynchronous update model is presented in Example \ref{ex:tenode_nd_zeroconsensus}. Similar to Section~\ref{sec:async_result}, we start by proving the convergence of series $\mathcal{F}^t$ by looking at convergence of $\mathcal{P}^t$ and $\mathcal{Q}^t.$ 

\begin{lemma}\label{lm:nd_P}
	Let $G = (V, E, \mathcal{W})$ be a matrix-weighted network with a spanning tree and negative definite edge weights. If agents update their states using \eqref{eq:SyncUpdate}, then $\lim\limits_{t \rightarrow \infty}\mathcal{P}^t$ converges almost surely to a zero matrix. 
\end{lemma}

See Appendix \ref{sec:app_d} for proof of Lemma \ref{lm:nd_P}. The convergence of $\mathcal{Q}^t$ can be proved on the lines of Lemma~\ref{lm:random_Q} and we skip the proof and state the result directly.
\begin{lemma}\label{lm:nd_Qconverge}
Let $G = (V, E, \mathcal{W})$ be a matrix-weighted network with a spanning tree and negative definite edge weights. If agents update their states using \eqref{eq:SyncUpdate}, then $\lim\limits_{t \rightarrow \infty}\mathcal{Q}^t$ converges almost surely to a zero matrix. 
\end{lemma}

By using the convergence of $\mathcal{P}^t$ and $\mathcal{Q}^t$ we prove the following result.
 
\begin{lemma}\label{lm:nd_small_eigen}
	Let $G = (V, E, \mathcal{W})$ be a matrix-weighted network with a spanning tree and negative definite edge weights. If agents update their states using \eqref{eq:SyncUpdate}, then $\lim\limits_{t \rightarrow \infty} \mathcal{F}^t$ converges almost surely to a zero matrix. 
\end{lemma}

See Appendix \ref{sec:app_e} for proof of Lemma \ref{lm:nd_small_eigen}. Now we are ready to present the main result of this section.
\begin{theorem}
	\label{thm:nd}
	Let $G = (V, E, \mathcal{W})$ be a matrix-weighted network with a spanning tree and negative definite edge weights. If agents update their state using \eqref{eq:Update}, then the network achieves zero consensus asymptotically almost surely for all initial states.
\end{theorem}
\begin{proof}
	The proof directly follows from the modified update rule~\eqref{eq:SyncUpdate} and Lemma~\ref{lm:nd_small_eigen}.
\end{proof}
\begin{example} \label{ex:tenode_nd_zeroconsensus}
    Consider the $10$ agents $4$-regular network shown in Figure \ref{fig:Tennodenw} with each edge weight being a negative definite matrix. Initial state vectors and the edge weight matrices are chosen randomly. State vectors evolution when agents update their states using \eqref{eq:Update_random} is shown in Figure \ref{fig:nd_zero}. Observe that all the agents converge to a zero vector thus achieving zero consensus.
\end{example}
\begin{figure*}
	\centering
	\subfigure[]{%
		\includegraphics[width=0.3\linewidth]{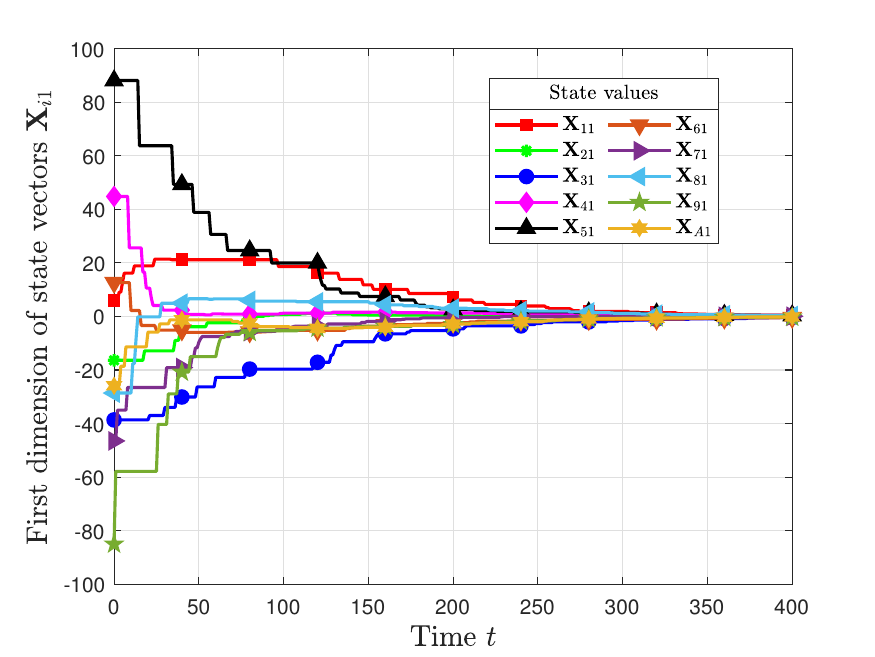}
		\label{fig:ND_ZC1}}
	\quad
	\subfigure[]{%
		\includegraphics[width=0.3\linewidth]{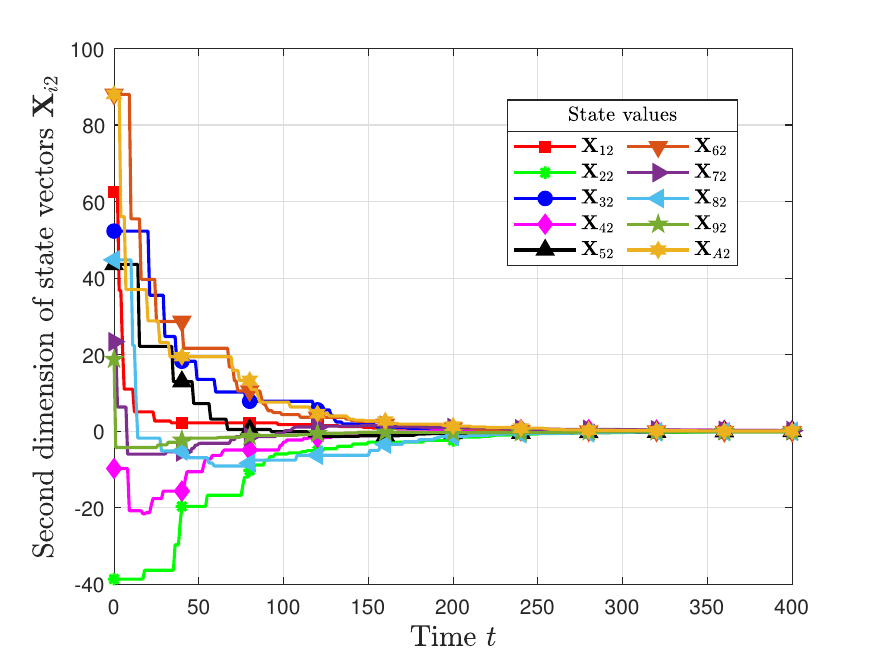}
		\label{fig:ND_ZC2}}
	\quad
	\subfigure[]{%
		\includegraphics[width=0.3\linewidth]{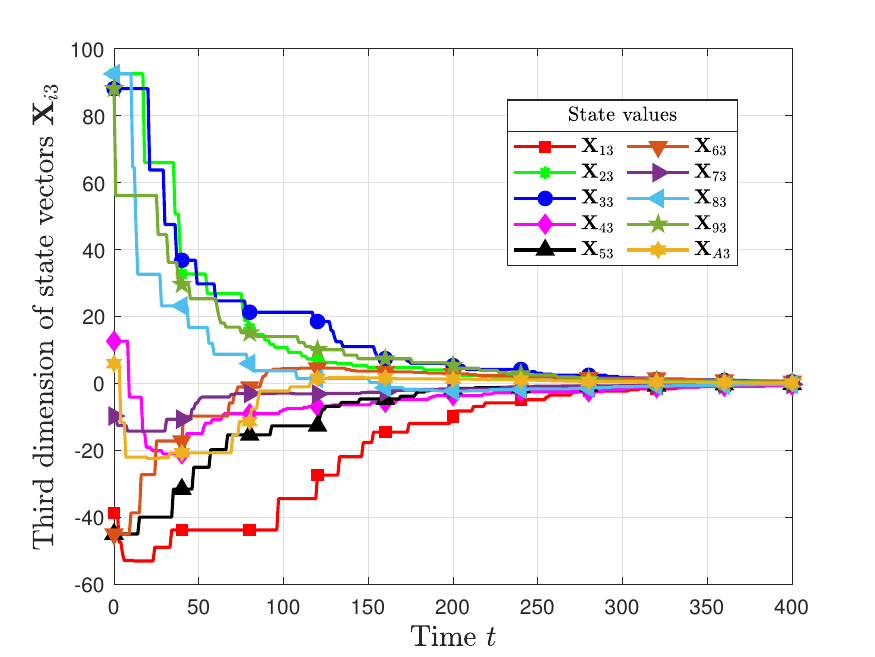}
		\label{fig:ND_ZC3}}
	\caption{Evolution of state vectors of an agent $i \in V$ for the network of Figure \ref{fig:Tennodenw} considering each edge weight as a negative definite matrix and state update with asynchronous update rule \eqref{eq:Update_random} showing zero consensus. (a) First dimension of state vectors. (b) Second dimension of state vectors.  (c) Third dimension of state vectors.}
	\label{fig:nd_zero}
\end{figure*}
\subsection{Structurally unbalanced networks}
\label{sec:unbalanced}

In this section we analyze vector consensus on structurally unbalanced networks. By the definition of a structurally unbalanced network there is at least an edge with positive definite and an edge with negative definite weight matrix. We start by proving that both $\mathcal{P}^t$ and $\mathcal{Q}^t$ matrices converge to zero matrix in this case thus leading to zero consensus. 
\begin{lemma}\label{lm:unbalanced_p}
Let $G = (V, E, \mathcal{W})$ be a structurally unbalanced matrix-weighted network with a spanning tree. If agents update their states using \eqref{eq:SyncUpdate}, then $\lim\limits_{t \rightarrow \infty}\mathcal{P}^t$ converges almost surely to a zero matrix.
\end{lemma}

See Appendix \ref{sec:app_f} for proof of Lemma \ref{lm:unbalanced_p}. 

The next lemma shows that $\lim\limits_{t \rightarrow \infty}\mathcal{Q}^t$ converges almost surely to a zero matrix when the network is structurally unbalanced.
\begin{lemma}\label{lm:unbalanced_q}
		Let $G = (V, E, \mathcal{W})$ be a structurally unbalanced matrix-weighted network with a spanning tree. If agents update their states using \eqref{eq:SyncUpdate}, then $\lim\limits_{t \rightarrow \infty}\mathcal{Q}^t$ converges almost surely to a zero matrix.
\end{lemma}

See Appendix \ref{sec:app_g} for proof of Lemma \ref{lm:unbalanced_q}. 

The proof of the following theorem follows by mimicking the proof of Theorem~\ref{th:random} by using Lemmas~\ref{lm:unbalanced_p},~\ref{lm:unbalanced_q}.
\begin{theorem}\label{thm:unbalanced}
		Let $G = (V, E, \mathcal{W})$ be a structurally unbalanced matrix-weighted network with a spanning tree. If agents update their states using \eqref{eq:Update}, then the network achieves zero consensus asymptotically almost surely for all initial states.
\end{theorem}
Combining the results of Theorem~\ref{th:bipartite} and Theorem~\ref{thm:unbalanced}, we obtain the following corollary.
\begin{corollary} \label{cor:necessary_sufficient}
	Let $G=(V,E,\mathcal{W})$ be a matrix-weighted network with a spanning tree. If agents update their states using \eqref{eq:Update}, then the network achieves bipartite consensus asymptotically almost surely for all initial states if and only if $G$ is structurally balanced.
\end{corollary}
\begin{example} \label{ex:tenode_zeroconsensus}
	Consider a structurally unbalanced network by making edge weights between agents 1, 3 and 6, 8 negative definite in Figure \ref{fig:SB}. State evolution on this network is shown in Figure \ref{fig:zero} when agents update their states using \eqref{eq:Update_random}. Observe that all agents converge to zero vector thus achieving zero consensus. 
\end{example}
\begin{figure*}
	\centering
	\subfigure[]{%
		\includegraphics[width=0.3\linewidth]{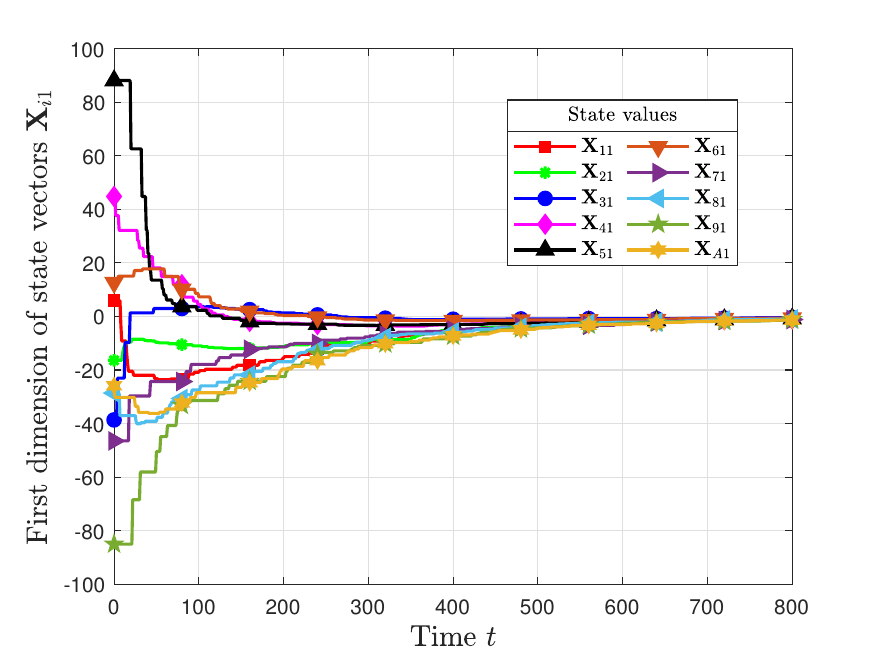}
		\label{fig:ZC1}}
	\quad
	\subfigure[]{%
		\includegraphics[width=0.3\linewidth]{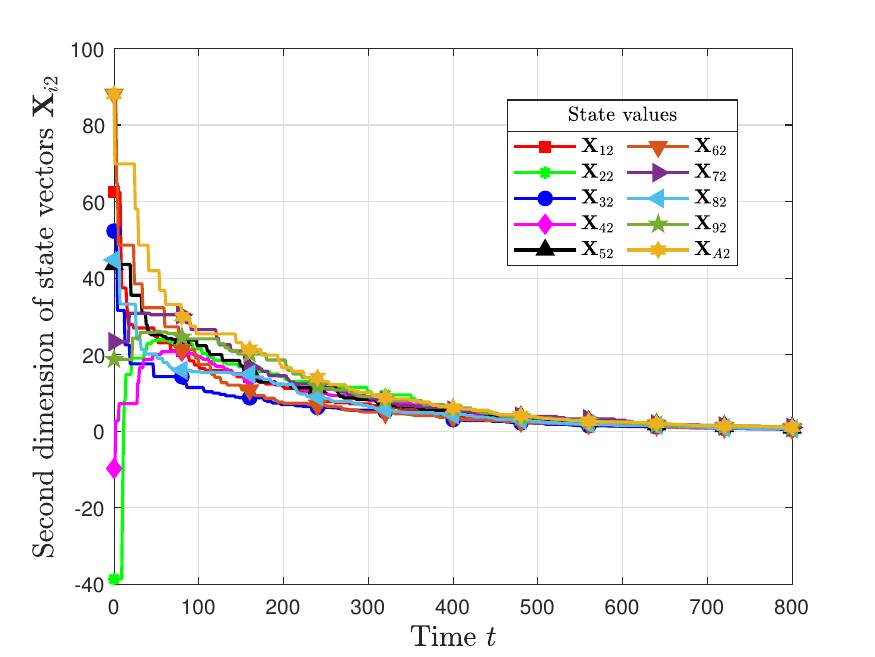}
		\label{fig:ZC2}}
	\quad
	\subfigure[]{%
		\includegraphics[width=0.3\linewidth]{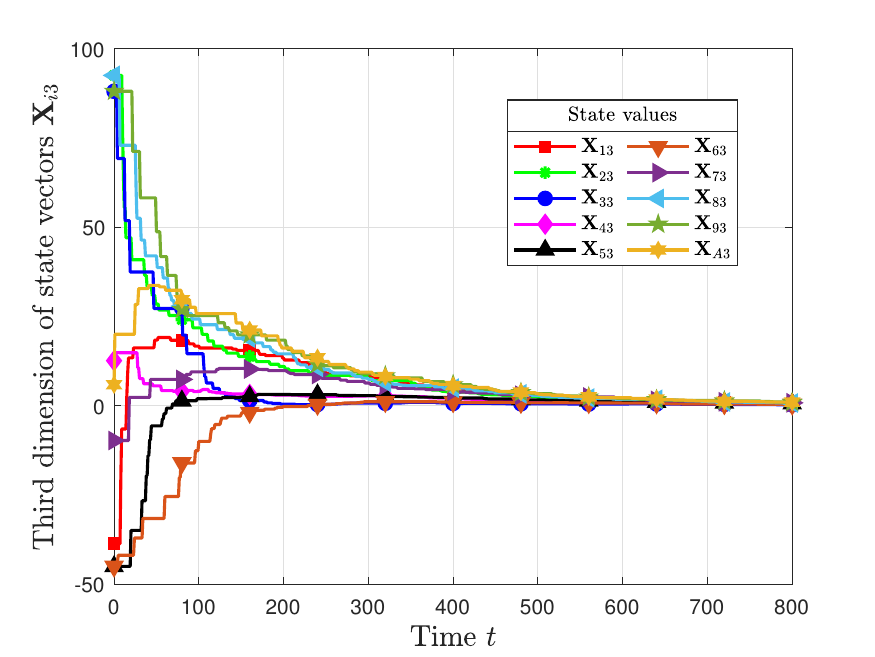}
		\label{fig:ZC3}}
	\caption{Evolution of state vectors of agents $i \in V$ for a structurally unbalanced network with asynchronous update rule \eqref{eq:Update_random} showing zero consensus. (a) First dimension of state vectors. (b) Second dimension of state vectors.  (c) Third dimension of state vectors.}
	\label{fig:zero}
\end{figure*}
%

\section{Numerical Results}
\label{sec:simulation}

In this section, we present numerical examples to illustrate some of the results presented in Sections \ref{sec:async_result}, \ref{sec:bipartite} and \ref{sec:zero}.

The following example shows that agents may converge to different state vectors on distinct sample paths in asynchronous update model despite having identical weight matrices and initial state vectors due to random sequence of update matrices.

\begin{example} \label{ex:tenode_globalconsensus1}
	Recall the network of Figure~\ref{fig:Tennodenw} and the global consensus achieved by agents using \eqref{eq:Update_random} shown in Figure~\ref{fig:global1}. Figure~\ref{fig:global2} shows the convergence of state vectors for the same network on a sample path other than that of Figure~\ref{fig:global1} with same weight matrices and initial state vectors. Observe that the state vectors converge to different consensus values in Figures~\ref{fig:global1} and \ref{fig:global2} due to random selection of updating agent at every time.
\end{example}
\begin{figure*}
	\centering
	\subfigure[]{%
		\includegraphics[width=0.3\linewidth]{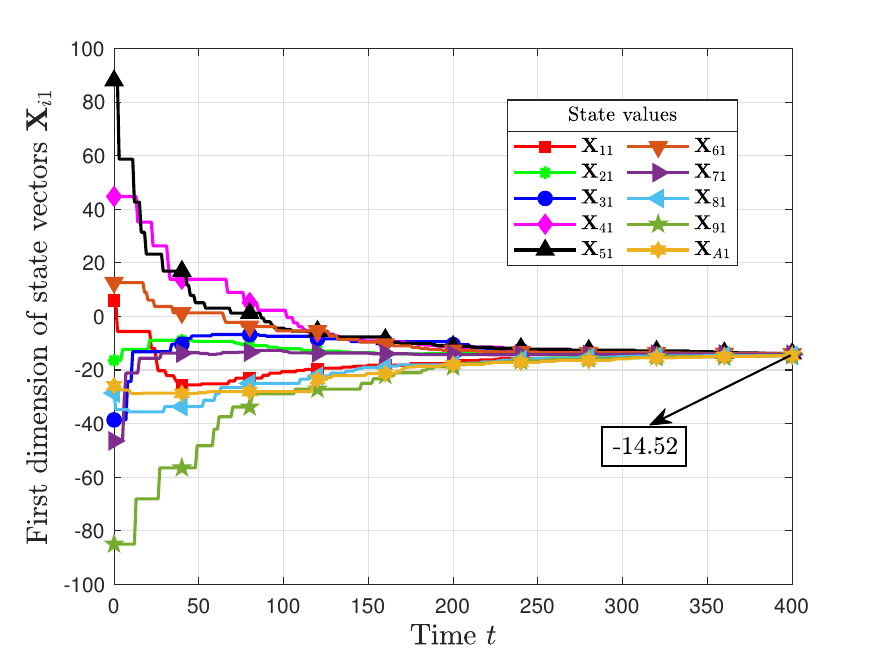}
		\label{fig:consensus1}}
	\quad
	\subfigure[]{%
		\includegraphics[width=0.3\linewidth]{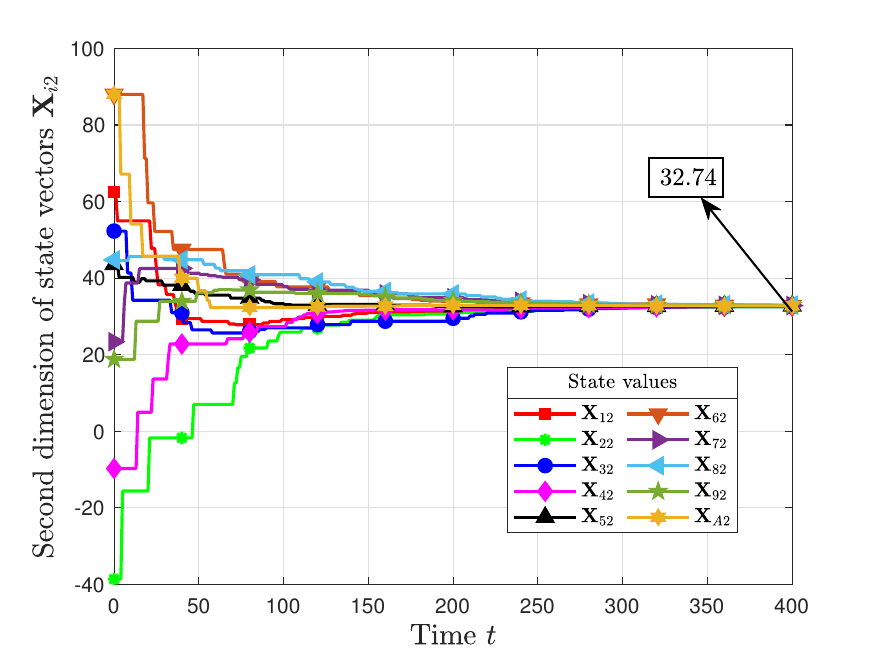}
		\label{fig:consensus2}}
	\quad
	\subfigure[]{%
		\includegraphics[width=0.3\linewidth]{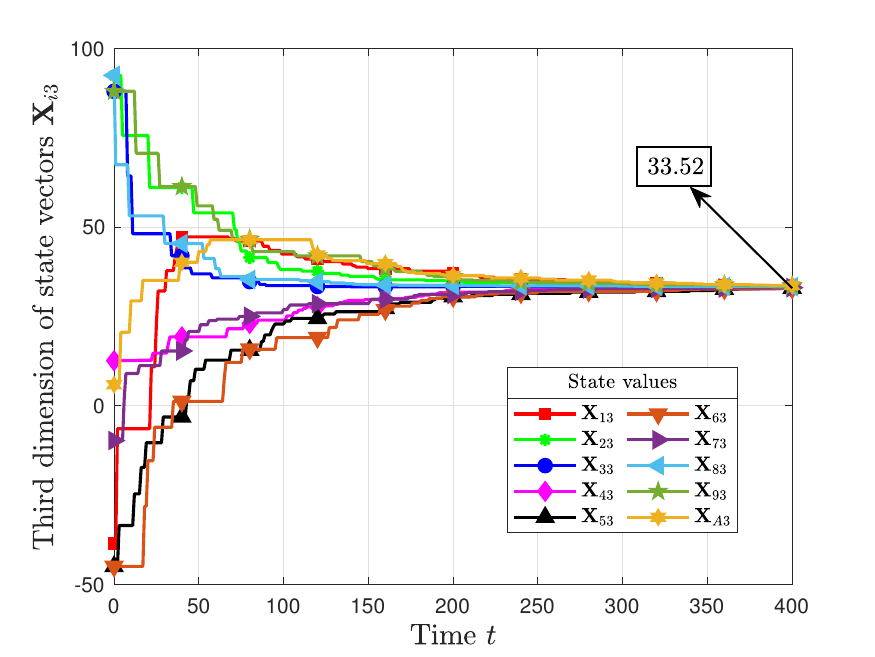}
		\label{fig:consensus3}}
	\caption{ Evolution of state vectors of an agent $i \in V$ for the network shown in Figure \ref{fig:Tennodenw} with asynchronous update rule \eqref{eq:Update_random} showing global consensus. (a) First dimension of state vectors. (b) Second dimension of state vectors. (c) Third dimension of state vectors.}
	\label{fig:global2}
\end{figure*}

The next example shows that the proper selection of step-size $\tau$ ensures the convergence of agents' states otherwise they become unbounded.
\begin{example}\label{ex:noconsensus}
	In Example \ref{ex:tenode_globalconsensus}, we select a step-size $\tau$ that satisfies the condition stated in Section \ref{sec:sync_model}, and it is observed that all agents in the network successfully achieve global consensus (see Figure~\ref{fig:global1}). Figure~\ref{fig:unbound} shows the state evolution for a step-size $\tau \notin \mathcal{R}$ when agents update their states using \eqref{eq:Update_random}. Observe that the state vectors become unbounded in this case due to non-convergence of $\prod_{k=0}^t \mathcal{F}(k)$ matrix.
\end{example}
\begin{figure*}
	\centering
	\subfigure[]{%
		\includegraphics[width=0.3\linewidth]{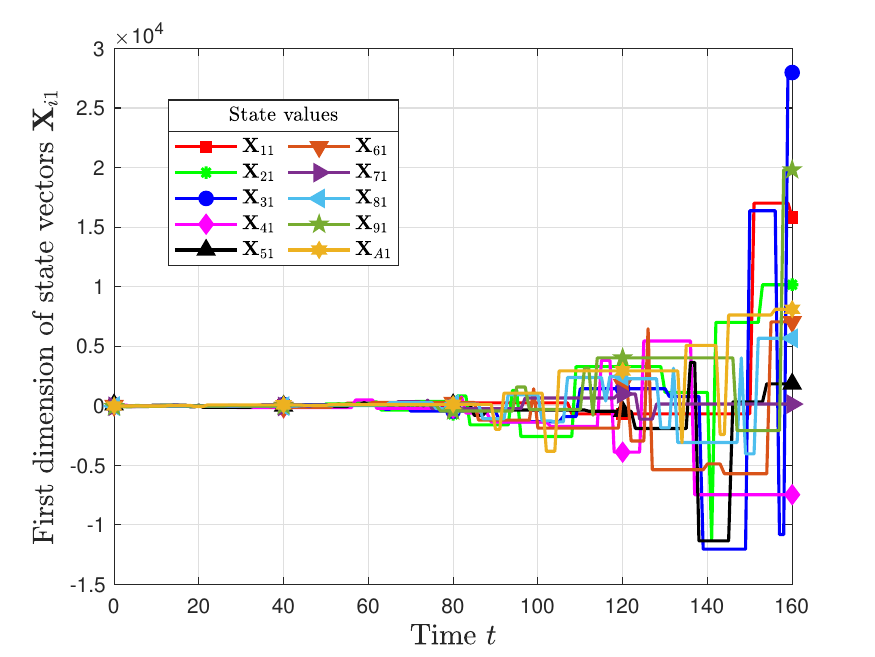}
		\label{fig:Noconsensus1}}
	\quad
	\subfigure[]{%
		\includegraphics[width=0.3\linewidth]{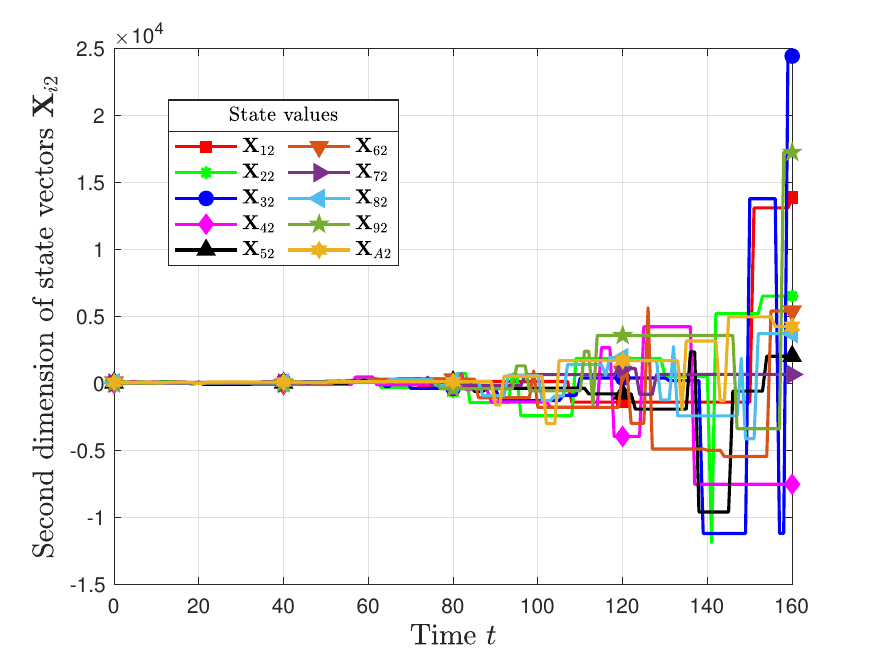}
		\label{fig:Noconsensus2}}
	\quad
	\subfigure[]{%
		\includegraphics[width=0.3\linewidth]{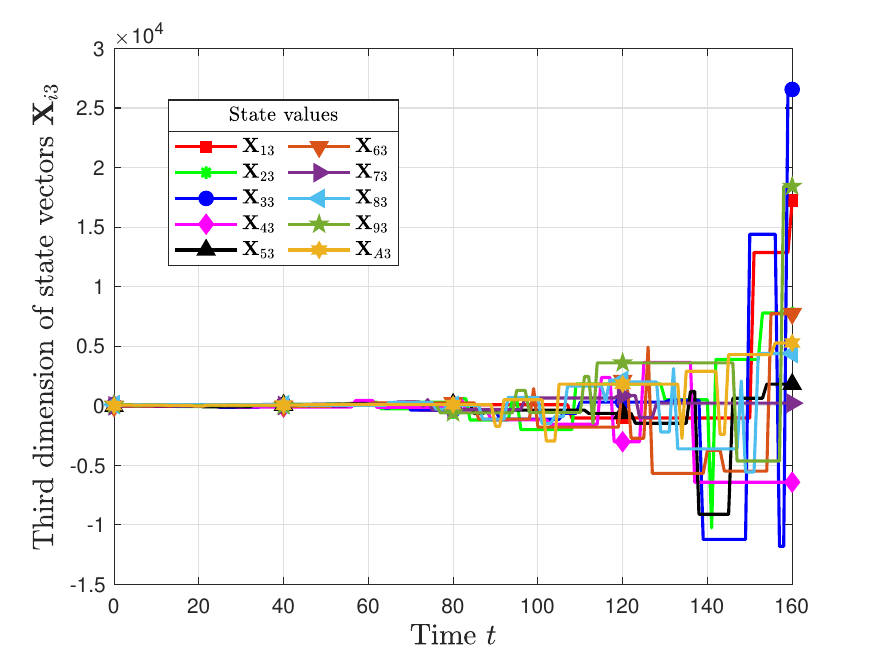}
		\label{fig:Noconsensus3}}
	\caption{ Evolution of state vectors for an agent $i \in V$ with asynchronous update rule \eqref{eq:Update_random} when $\tau \notin \mathcal{R}$ showing divergence of state vectors. (a) First dimension of state vectors. (b) Second dimension of state vectors. (c) Third dimension of state vectors.}
	\label{fig:unbound}
\end{figure*}

Now we show the convergence of agents' states in a 5 agents directed network showing that our results hold true for both undirected and directed networks with a spanning tree.

\begin{example} \label{ex:fivenode_globalconsensus}
	Observe that the directed network of Figure \ref{fig:5node} is not strongly connected but has a spanning tree. The evolution of the state vectors for every agent when agents update their states using \eqref{eq:Update_random} (see Figure~\ref{fig:5node_consensus}) shows that agents achieve global consensus. This shows that our result holds for directed graphs as well thus showing the robustness of our results.
\end{example}
\begin{figure*}
	\centering
	\subfigure[]{%
		\includegraphics[width=0.3\linewidth]{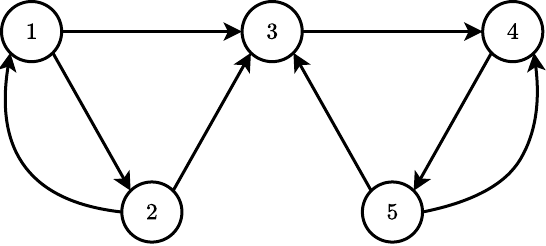}
		\label{fig:5node}}
	\quad
	\subfigure[]{%
		\includegraphics[width=0.3\linewidth]{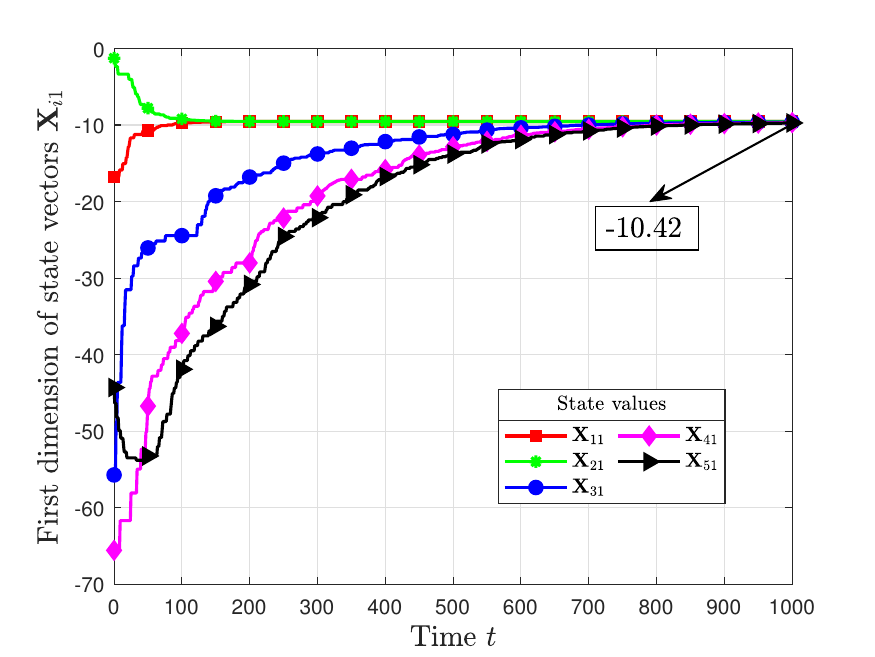}
		\label{fig:5node_dim1}}
	\quad
	\subfigure[]{%
		\includegraphics[width=0.3\linewidth]{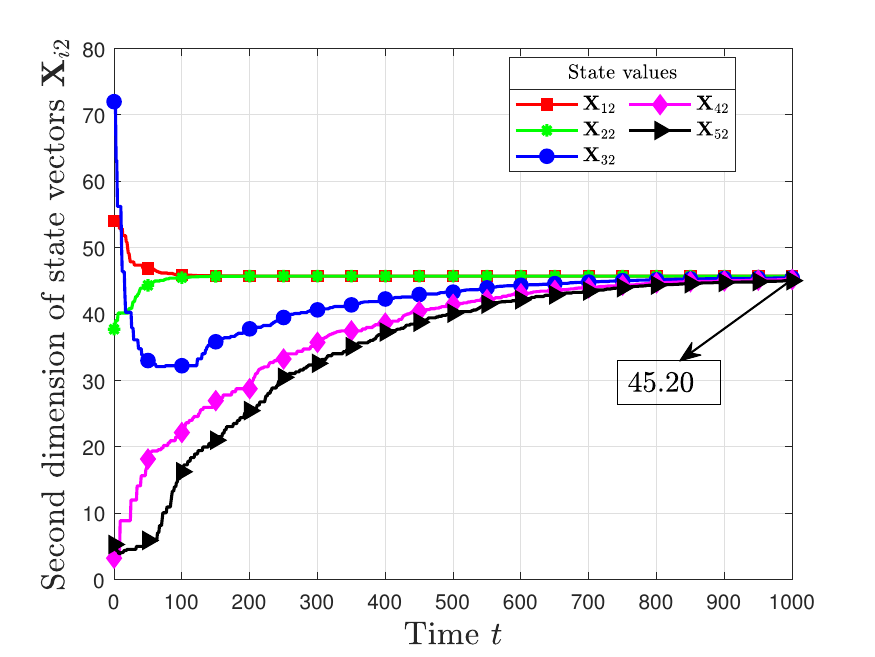}
		\label{fig:5node_dim2}}
	\caption{ Evolution of state vectors of an agent $i \in V$ for the network shown in Figure \ref{fig:5node} with asynchronous update rule \eqref{eq:Update_random} showing global consensus. (a) A directed network with $5$ agents containing spanning tree. (b) First dimension of state vectors. (c) Second dimension of state vectors.}
	\label{fig:5node_consensus}
\end{figure*}
Examples~\ref{ex:tenode_globalconsensus}-\ref{ex:fivenode_globalconsensus} consider specific networks generated by fixed edges. The following example shows that the results hold true for any randomly generated network as well. 
\begin{example}
	We consider a random geometric graph (RGG) \cite{penrose03} in which agents are placed in a metric space, typically the Euclidean plane, and edges are formed based on the distance between vertices\footnote{There are various ways to generate RGG networks. We generate the network where edges are chosen based on Euclidean distance. See \cite{bollobas98} for more details.}. We consider a RGG network, denoted by $G(200,0.4),$ in which $200$ agents are distributed uniformly at random in a unit square and an edge is between two agents if their Euclidean distance is less than or equal to the threshold $0.4.$ We make sure that our network $G(200,0.4)$ has a spanning tree. The evolution of state vectors of agents $\{2,97,160,198\} \in V,$ when agents update their states using \eqref{eq:Update_random} is shown in Figure \ref{fig:RGG_consensus}. Observe that all agents reach to same vector thus achieving global consensus.
\end{example}
\begin{figure*}
	\centering
	\subfigure[]{%
		\includegraphics[width=0.3\linewidth]{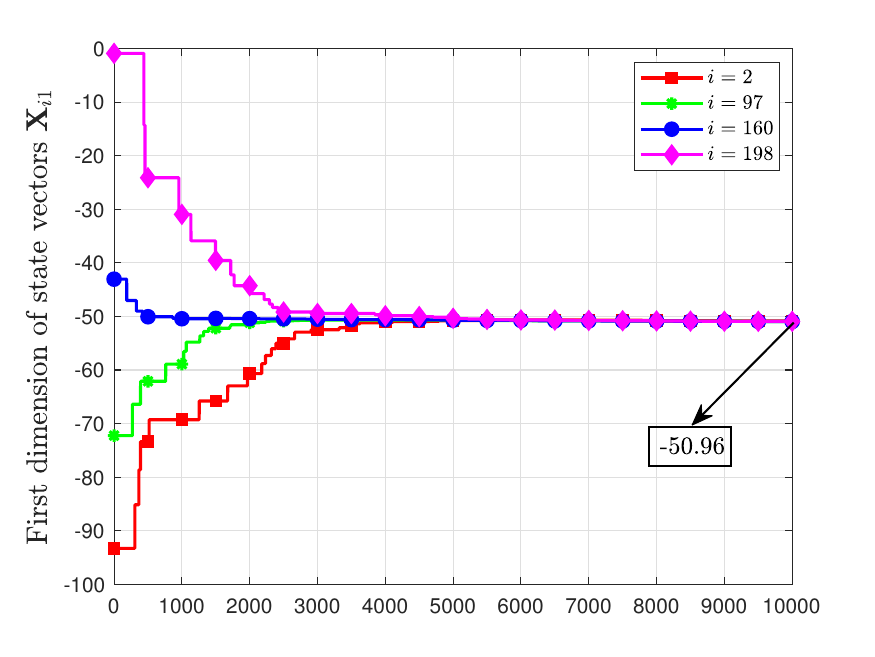}
		\label{fig:rgg_dim1}}
	\quad
	\subfigure[]{%
		\includegraphics[width=0.3\linewidth]{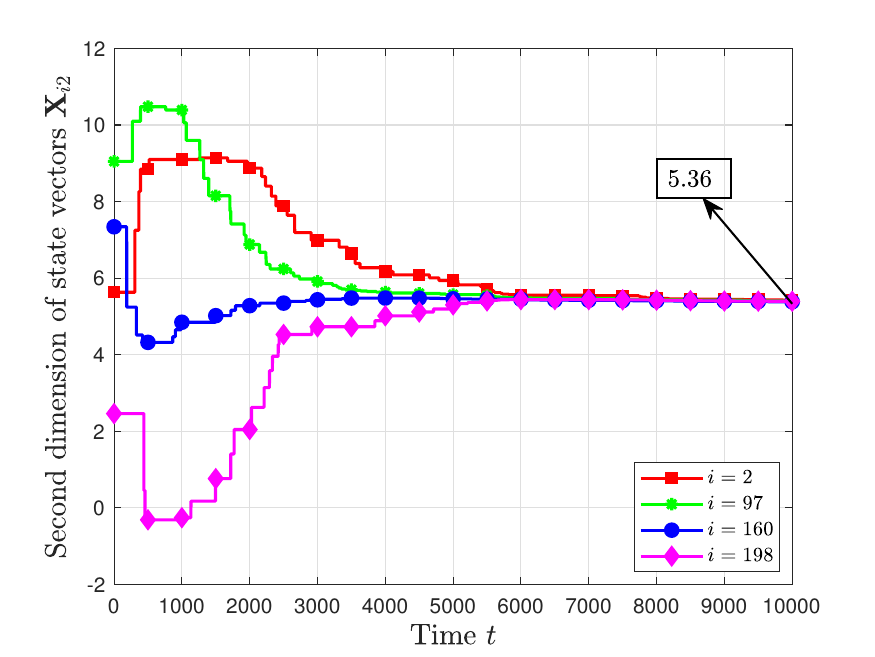}
		\label{fig:rgg_dim2}}
	\quad
	\subfigure[]{%
		\includegraphics[width=0.3\linewidth]{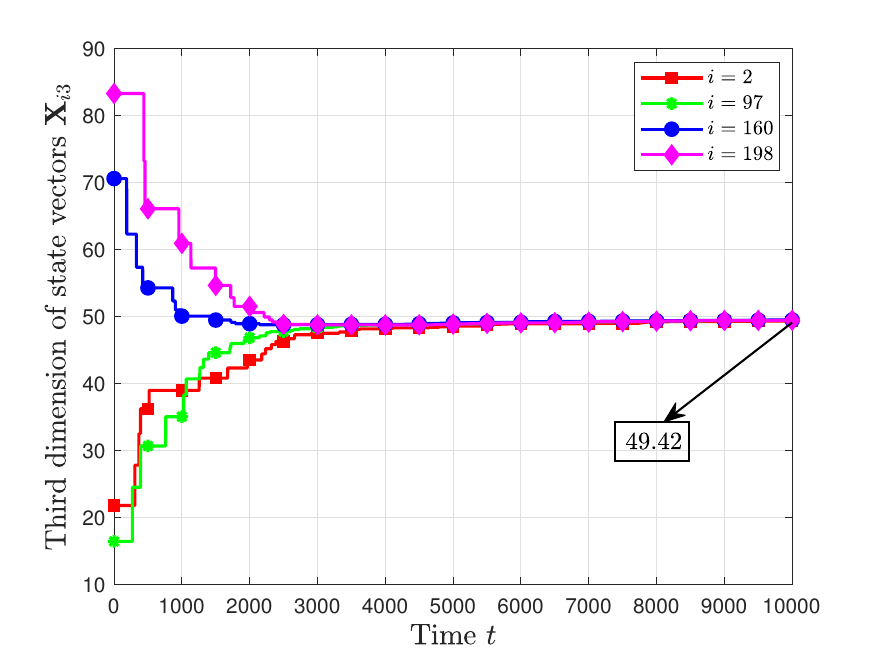}
		\label{fig:rgg_dim3}}
	\caption{ Evolution of state vectors of agents $\{2,97,160,198\}  \in V$ for the random geometric network $G(200,0.4)$ with asynchronous update rule \eqref{eq:Update_random} showing global consensus. (a) First dimension of state vectors. (b) Second dimension of state vectors. (c) Third dimension of state vectors.}
	\label{fig:RGG_consensus}
\end{figure*}
%

\section{Conclusion}
\label{sec:conclusion}

In this work, we studied the vector consensus over discrete-time matrix-weighted multi-agent networks using stochastic matrix convergence theory. In particular, we showed that the network achieves the global consensus almost surely for both synchronous and asynchronous update models when all edge weights are positive definite. The existing literature only shows the convergence in expectation for asynchronous update model.
We also considered the situation when the agents' interactions are either cooperative (denoted by positive definite edge weight) or competitive (edge weight being negative definite). We showed that the network achieves bipartite consensus almost surely if and only if the network is structurally balanced.Network achieves zero consensus almost surely when network is structurally unbalanced.  While proving our main results, we proved convergence of product of sum of nonhomogenous matrices which is of independent interest in the matrix convergence theory. Interesting future directions for exploration are the vector consensus with communication delays and quantized communication.
%
\bibliographystyle{IEEEtran}
\bibliography{Vector}

\input{Appendix.tex}
\end{document}

%% file: Appendix.tex
\appendices
\section{Proof of Lemma~7}
\label{sec:app_a}
\begin{proof}
	By the properties of $Q_{ij}$ matrices, the row sum of $\mathcal{Q}(k), \forall k$ is zero and $0 < ||\mathcal{Q}(k)|| <1.$ Let $||\mathcal{Q}(k)||_+ = \max(||\mathcal{Q}(k)||,1)$ and $||\mathcal{Q}(k)||_- = \min(||\mathcal{Q}(k)||,1).$ Observe that $\sum_{k=0}^\infty (||\mathcal{Q}(k)||_+ -1)$ converges and $\sum_{k=0}^\infty (1- ||\mathcal{Q}(k)||_-)$ diverges. By Theorem~6.3 of \cite{Hartfiel02}, $\lim\limits_{t \rightarrow \infty} \prod_{k=0}^t\mathcal{Q}(k) =\mathbf{0}.$ Note that this argument does not need the existence of spanning tree on the sample path thus the result holds true for all the sample paths of $\Omega.$ 
\end{proof}
\section{Proof of Lemma~8}
\label{sec:app_b}
\begin{proof}
	On any sample path $\omega \in \Omega',$ recall that $\prod_{k=0}^t\mathcal{F}(k) = \prod_{k=0}^t(\mathcal{P}(k)+\mathcal{Q}(k)).$ We skip the notation $\omega.$ Note that the matrix norm of $\mathcal{P}(k)$ is $||\mathcal{P}(k)|| =1, \forall k, \lim\limits_{t \rightarrow \infty}\prod_{k=0}^t\mathcal{P}(k)$ exists by Lemma \ref{lm:random_P} and $\lim\limits_{t \rightarrow \infty} \prod_{k=0}^t\mathcal{Q}(k) \rightarrow \mathbf{0}$ by Lemma \ref{lm:random_Q}. Let $\mathcal{J}_t=\prod_{k=0}^t(\mathcal{P}(k)+\mathcal{Q}(k)), \forall t \geq t'(\omega)$ be a matrix sequence. We will show that the sequence $\mathcal{J}_t$ converges by proving that $\mathcal{J}_t$ is a Cauchy sequence. Let for some $t > r > t'(\omega),$ $\mathcal{J}_{t,r}=\prod_{l=0}^{t-r}\mathcal{P}(l)\prod_{k=0}^t(\mathcal{P}(k)+\mathcal{Q}(k)).$ By triangular inequality of the matrix norm, for some $t, s > r,$
	\begin{equation}
		||\mathcal{J}_t-\mathcal{J}_s|| \leq ||\mathcal{J}_t-\mathcal{J}_{t,r}||+||\mathcal{J}_{t,r}-\mathcal{J}_{s,r}||+||\mathcal{J}_{s,r}-\mathcal{J}_s||. \label{eq:J_expand_random}
	\end{equation}
	By using the fact that $||\mathcal{P}(k)|| =1, \forall k$ from Lemma~\ref{lm:random_P} and $\lim\limits_{t \rightarrow \infty} \prod_{k=0}^t\mathcal{Q}(k) \rightarrow \mathbf{0}$ from Lemma \ref{lm:random_Q}, one can show that for sufficiently large $t,s$ and $r,$ $||\mathcal{J}_t-\mathcal{J}_{t,r}|| \leq \epsilon/3, ||\mathcal{J}_{s,r}-\mathcal{J}_s|| \leq \epsilon/3$ on the lines of the proof of Lemma~\ref{lm:cauchy} for some constant $\epsilon>0.$ By using the result $\lim\limits_{t \rightarrow \infty}\prod_{k=0}^t\mathcal{P}(k)$ exists by Lemma \ref{lm:random_P}, one can also show that $||\mathcal{J}_{t,r}-\mathcal{J}_{s,r}|| \leq \epsilon/3.$  Hence $||\mathcal{J}_t-\mathcal{J}_s|| < \epsilon$ which implies, $\mathcal{J}_t, \forall t \geq t'(\omega)$ is a Cauchy sequence. As the matrix product space with norm $||.||$ is complete, this shows that the Cauchy sequence $\mathcal{J}_t, \forall t \geq t'(\omega)$ converges. This proves that for any $\omega \in \Omega',$ $\lim\limits_{t \rightarrow \infty}\prod_{k=0}^t\mathcal{F}(k)$ exists.
\end{proof}
\section{Proof of Lemma~9}
\label{sec:app_c}
\begin{proof}
	For every $t \geq t'(\omega),$ let $\prod_{k=0}^t\mathcal{F}(k)=\mathcal{A}_t+\mathcal{B}_t$ with $\mathcal{A}_t$ having similar structure to that of $\mathcal{P}$ in \eqref{eq:P} and $\mathcal{B}_t$ having similar structure to that of $\mathcal{Q}$ in \eqref{eq:Q}. Here $P_i^t$ is an $i^{th}$ diagonal block matrix in $\mathcal{A}_t$ similar to that of $P_i$ in $\mathcal{P}$ and $Q_{ij}^t$ is an $(i,j)^{th}$ block matrix in $\mathcal{B}_t$ similar to that of $Q_{ij}$ in $\mathcal{Q}.$ Using the arguments made in Lemma~\ref{lm:random_P}, $P_i^t, \forall i \in \{1,2,\ldots,d\}$ in $\mathcal{A}_t$ is a stochastic matrix with positive diagonal elements and the induced graph of $P_i^t, \forall i$ has a spanning tree. By Corollary~$3.5$ of \cite{ren2005consensus}, $P_i^t$ has algebraic multiplicity of one for eigenvalue $\lambda = 1$ and its all other eigenvalues are less than one. Let $v_i=k\mathbf{1} \in \mathbb{R}^n$ be the eigenvector of $P_i^t$ for some $i$ corresponding to $\lambda=1$ for some positive constant $k.$ Let $\overrightarrow{V_i} =[\mathbf{0}_n,\cdots,\mathbf{0}_n,v_i,\mathbf{0}_n,\cdots,\mathbf{0}_n]^T \in \mathbb{R}^{nd}$ with $v_i$ at the $i^{th}$ position. It is easy to observe that $\overrightarrow{V_i}$ is an eigenvector of $\mathcal{A}_t$ corresponding to $\lambda=1.$ Hence, the unit eigenvalue of $\mathcal{A}_t$ has algebraic multiplicity of $d$ with one eigenvector coming from each $P_i^t,\forall i.$
	
	The matrix $\prod_{k=0}^t\mathcal{F}(k), \forall t \geq t'(\omega)$ has $P_i^t$ matrices in the diagonal and $Q_{ij}^t$ matrices in its off-diagonal blocks. Each row sum of $Q_{ij}^t$ matrix is zero. Hence $\overrightarrow{V_1},\cdots,\overrightarrow{V_d}$ are also eigenvectors of $\prod_{k=0}^t\mathcal{F}(k), \forall t \geq t'(\omega)$ corresponding to $\lambda=1.$ Thus $\prod_{k=0}^t\mathcal{F}(k), \forall t \geq t'(\omega)$ has an eigenvalue $\lambda=1$ with $d$ multiplicity. Observe that the vectors $\overrightarrow{V_1},\cdots,\overrightarrow{V_d}$ are linearly independent hence the geometric multiplicity of unit eigenvalue of $\prod_{k=0}^t\mathcal{F}(k), \forall t \geq t'(\omega)$ is also $d.$
\end{proof}
\section{Proof of Lemma~11}
\label{sec:app_d}
\begin{proof}
	Recall the definition of $\mathcal{P}$ and $P_i$ from \eqref{eq:P} and \eqref{eq:Pi}. Note that when all the interaction weight matrices are negative definite, all the sign functions $\text{sgn}(W_{km})$ are negative and hence $P_i$s are not stochastic matrices. Let $z$ be an eigenvalue of $P_i$ for some $i \in \{1,\ldots,d\}.$ By Gershgorin theorem \cite{horn2012matrix} we know that,
	\begin{align*}
		\abs*{z-\left( 1+\tau \sum_{l \in \mathcal{N}_k}{W_{kl}^{(i,i)}} \right)}  & \stackrel {\text{(a)}} {\leq} \tau \sum_{l \in \mathcal{N}_k}{\mid W_{kl}^{(i,i)} \mid}\\
		& \stackrel {\text{(b)}} {\leq} -\tau \sum_{l \in \mathcal{N}_k}{ W_{kl}^{(i,i)}}. 
	\end{align*}
	Here the inequality (a) follows from Gershgorin Theorem and the inequality (b) follows from the negative definiteness of $W_{kl}$ matrices. Thus by the definition of step size $\tau,$ the eigenvalue can be lower bounded as $z \geq 1+\tau \sum_{l \in \mathcal{N}_k}{W_{kl}^{(i,i)}} + \tau \sum_{l \in \mathcal{N}_k}{W_{kl}^{(i,i)}} >0.$
	Similarly eigenvalues can be upper bounded by $z \leq 1+\tau \sum_{l \in \mathcal{N}_k}{W_{kl}^{(i,i)}} - \tau \sum_{l \in \mathcal{N}_k}{W_{kl}^{(i,i)}} = 1.$ Thus all the eigenvalues of matrix $P_i$ are $0 < z \leq 1, \forall i \in \{1,\ldots,d\}.$ 
	
	Now we will prove that any eigenvalue of $P_i$ can not equal to $1$ by contradiction. Let us assume that $P_i$ has a unity eigenvalue $z$ with eigenvector $Y$ such that $(P_i-\mathbf{I})Y =\mathbf{0}.$ This implies that $Y$ is in the null space of matrix $P_i-I.$ Let $v_1,v_2,\ldots,v_n$ be the columns of $(P_i-\mathbf{I})$ and $v:=a_1v_1+\ldots+a_nv_n$ where $a_i \in \mathbb{R}, \forall i$ are some constants. Column vectors $v_1,\ldots,v_n$ are linearly independent if $v=\mathbf{0}$ only when $a_i=0, \forall i.$ The first dimension of $v$ can be zero iff $a_2=a_3=\ldots=a_n=-a_1.$ Similarly, the $n^{th}$ dimension of $v$ is $0$ iff $a_1=a_2=\ldots=a_{n-1}=-a_n.$ Hence $v=\mathbf{0}$ iff $a_i=0, \forall i$ implying that all columns of $P_i-\mathbf{I}$ are linearly independent. Thus the rank of $(P_i-\mathbf{I})$ is $n.$ Using the rank-nullity theorem \cite{horn2012matrix}, the nullity of $(P_i-\mathbf{I})$ is $0$ which implies that the null space of $(P_i - \mathbf{I})$ is empty, i.e., $\nexists$ a non-zero $Y$ such that $(P_i - \mathbf{I})Y=\mathbf{0}.$ Thus, $z$ can not be equal to one. As $|z| <1,$ by Lemma~\ref{lm:zeroconvergence} it follows that $\lim\limits_{t \rightarrow \infty} P_i^t =\mathbf{0}.$
\end{proof}
\section{Proof of Lemma~13}
\label{sec:app_e}
\begin{proof}
	First we prove that $1$ is not an eigenvalue of $\mathcal{F}$ by contradiction. Let us assume that $\lambda =1$ is an eigenvalue of $\mathcal{F}$ with eigenvector $Y.$ Then $\left(\mathcal{F} - \mathbf{I} \right) Y=\mathbf{0}.$ Let $u_1,u_2,\ldots,u_{nd}$ be the column vectors of $\left(\mathcal{F} - \mathbf{I} \right).$ Consider the linear combination of column vectors as $u=a_1u_1+a_2u_2+\ldots+a_{nd}u_{nd}$ where $a_i \in \mathbb{R}, \forall i$ is a scalar. Column vectors $u_1,u_2,\ldots,u_{nd}$ are linearly independent if $u=\mathbf{0}$ can only be satisfied by $a_i=0, ~\forall i.$ The $k^{th}$ row of $u$ is zero if $-a_{(k-1)n+1}=a_{(k-1)n+1+j}, k=1,\ldots,d$ and $j=1,\ldots,(n-1).$ Now, it is easy to verify that all the rows of $u$ are zero, i.e., $u=\mathbf{0}$ iff $a_i=0,~\forall i.$ Hence all the columns of $\left(\mathcal{F} - \mathbf{I} \right)$ are linearly independent to each other and its rank is $n.$ Using the rank-nullity theorem \cite{horn2012matrix}, the nullity of $\left(\mathcal{F} - \mathbf{I} \right)$ is $0$ and there does not exist a non-zero $Y$ such that $\left(\mathcal{F} - \mathbf{I} \right)Y=\mathbf{0} \Rightarrow |\lambda| \neq 1.$
	
	The argument that no eigenvalue of $\mathcal{F}$ has magnitude greater than one follows on the same line of the proof of Lemma~\ref{lm:PQ_eigen_less_random}. Hence the matrix $\mathcal{F}$ has all the eigenvalues with magnitude less than one. The result now directly follows by applying Lemma~\ref{lm:zeroconvergence}.
\end{proof}
\section{Proof of Lemma~14}
\label{sec:app_f}
\begin{proof}
	We prove the result by using Lemma~\ref{lm:zeroconvergence} by showing that magnitude of all eigenvalues of $P_i,\forall i$ (recall $P_i$ from \eqref{eq:Pi}) are strictly less than one. Let $\mathcal{N}_{k}^T$ be the set of in-neighbors of agent $k$ with positive definite edge weight matrix and $\mathcal{N}_{k}^{NT}$ be the set of in-neighbors with negative definite edge weight matrix. Then, by \eqref{eq:Pi} and Gershgorin Theorem, every eigenvalue $z$ of $P_i$ can be written as:
	\begin{align*}
		&\abs*{z -\left(1-\tau \sum_{l \in \mathcal{N}_k^T}{W_{kl}^{(i,i)}}+\tau \sum_{l \in \mathcal{N}_k^{NT}}{W_{kl}^{(i,i)}} \right)} \\
		& \stackrel {\text{(a)}} {\leq} \tau \sum_{l \in \mathcal{N}_k^T}{\mid W_{kl}^{(i,i)} \mid} + \tau \sum_{l \in \mathcal{N}_k^{NT}}{\mid W_{kl}^{(i,i)} \mid} \\
		& \stackrel {\text{(b)}} {\leq} \tau \sum_{l \in \mathcal{N}_k^T} W_{kl}^{(i,i)} - \tau \sum_{l \in \mathcal{N}_k^{NT}} W_{kl}^{(i,i)}.
	\end{align*}
    Here the inequality (a) follows from Gershgorin Theorem and the inequality (b) is because $W_{kl}^{(i,i)} \leq 0$ whenever $W_{kl}$ is a negative definite matrix. Thus, $0<1-2 \tau \sum_{l \in \mathcal{N}_k^T}{W_{kl}^{(i,i)}}+2 \tau \sum_{l \in \mathcal{N}_k^{NT}}{W_{kl}^{(i,i)}} \leq |z| \leq 1.$
    This shows that every eigenvalue of $P_i$ is less than or equal to one. 
    
    Now we prove that $1$ is not an eigenvalue of $P_i$ by contradiction. Let $\lambda =1 $ be an eigenvalue of $P_i$ with eigenvector $Y$ and $(P_i - \mathbf{I})Y=\mathbf{0}.$ Let $y_1,\ldots,y_n$ be the column vectors of $P_i - \mathbf{I}.$ Consider the linear combination of column vectors as $y=b_1y_1+\ldots+b_ny_n$ where $b_i,\forall i$ is a scalar. Column vectors $y_1,\ldots,y_n$ are linearly independent if $y=\mathbf{0}$ iff $b_i=0~\forall i.$ The first row of $y$ is zero if $b_1=b_k ~\forall k \in \mathcal{N}_1^T$ and $b_1=-b_k ~\forall k \in \mathcal{N}_1^{NT}.$ Similarly, the $l^{th}$ row of $y$ is zero if $b_l=b_k~\forall k \in \mathcal{N}_l^T$ and $b_l=-b_k ~\forall k \in \mathcal{N}_l^{NT}.$ As the network is structurally unbalanced all these equalities will be true iff $b_i=0, \forall i.$ Hence, the rank of $P_i - \mathbf{I}$ is $n.$ Using the rank-nullity theorem \cite{horn2012matrix}, the null space of $P_i - \mathbf{I}$ is empty, i.e., $\nexists$ a non-zero $Y$ such that $(P_i - \mathbf{I})Y=\mathbf{0} \Rightarrow \lambda \neq 1.$ This contradicts our assumption hence all eigenvalues of $P_i$ are less than one.
   \end{proof}
\section{Proof of Lemma~15}
\label{sec:app_g}
\begin{proof}
	By Gershgorin circle theorem, every eigenvalue $z$ of $\mathcal{Q}$ is in the disc $\mid z- \alpha \mid \leq \gamma,$ where $\alpha = 0$ is the principal diagonal element of $\mathcal{Q}$ and $\gamma$ is the absolute sum of all the off diagonal elements of a row of $\mathcal{Q}.$ Thus,
	\begin{align*}
		\gamma &= \sum_{j \in \mathcal{N}_i} \mid -\tau \sum_{m \in \mathcal{N}_k^T}{W_{km}^{(i,i)}}+\tau \sum_{m \in \mathcal{N}_k^{NT}}{W_{km}^{(i,i)}} \mid \\
		&+ \sum_{j \in \mathcal{N}_i} \sum_{m \in \mathcal{N}_k} \mid \tau W_{km}^{(i,j)} \mid\\
		& \stackrel {\text{(a)}} {\leq} \tau \sum_{j \in \mathcal{N}_i} \sum_{m \in \mathcal{N}_k^T} \mid W_{km}^{(i,i)} \mid + \tau \sum_{j \in \mathcal{N}_i} \sum_{m \in \mathcal{N}_k^{NT}} \mid W_{km}^{(i,i)} \mid\\
		&+ \tau \sum_{j \in \mathcal{N}_i} \sum_{m \in \mathcal{N}_k} \mid  W_{km}^{(i,j)} \mid \\
		& = \tau \sum_{j \in \mathcal{N}_i} \sum_{m \in \mathcal{N}_k} \mid W_{km}^{(i,i)} \mid + \tau \sum_{j \in \mathcal{N}_i} \sum_{m \in \mathcal{N}_k} \mid  W_{km}^{(i,j)} \mid\\
		& = 2 \tau \sum_{j \in \mathcal{N}_i} \sum_{m \in \mathcal{N}_k} \mid  W_{km}^{(i,j)} \mid \\
		& \stackrel {\text{(b)}} {<} 1.		 
	\end{align*}
	Here the inequality (a) is due to triangular inequality, and the inequality (b) follows from the definition of step size $\tau.$ Thus the magnitude of all eigenvalues of $\mathcal{Q}$ are less than one and the result follows from Lemma~\ref{lm:zeroconvergence}.
\end{proof}